\documentclass[aps,prl,showpacs,superscriptaddress,twocolumn,final,nofootinbib,10pt]{revtex4}
\usepackage{amsmath,amsfonts,amsthm}
\usepackage{hyperref}
\usepackage{graphicx,color}
\usepackage{multirow}
\usepackage{braket}
\usepackage{bbm}
\usepackage{textcomp}
\usepackage{soul}
\usepackage{accents}

\newcommand{\ie}{i.e.}
\newcommand{\eg}{e.g.}
\newcommand{\sthat}{\textrm{s.t.}}

\newtheorem{proposition}{Proposition}

\DeclareMathOperator{\tr}{tr}
\DeclareMathOperator{\Proj}{proj}

\newcommand\T{\rule{0pt}{3.2ex}}
\newcommand\B{\rule[-1.4ex]{0pt}{0pt}}

\newcommand{\nn}{n}

\newtheorem{definition}{Definition}

\definecolor{nred}{rgb}{0.9,0.2,0.2}
\definecolor{nblack}{rgb}{0,0,0}
\definecolor{nblue}{rgb}{0.2,0.2,0.7}
\definecolor{ngreen}{rgb}{0.2,0.6,0.2}

\DeclareMathSymbol{\fixwidehatsym}{\mathord}{largesymbols}{"62}
\newcommand\lowerwidehatsym{%
  \text{\smash{\raisebox{-1.3ex}{%
    $\fixwidehatsym$}}}}
\newcommand\fixwidehat[1]{%
  \mathchoice
    {\accentset{\displaystyle\lowerwidehatsym}{#1}}
    {\accentset{\textstyle\lowerwidehatsym}{#1}}
    {\accentset{\scriptstyle\lowerwidehatsym}{#1}}
    {\accentset{\scriptscriptstyle\lowerwidehatsym}{#1}}
}

\begin{document}
\title{Device-independent entanglement quantification and related applications}

\author{Tobias~Moroder}
\affiliation{Naturwissenschaftlich-Technische Fakult\"at, Universit\"at Siegen, Walter-Flex-Str.~3, D-57068 Siegen, Germany}
\author{Jean-Daniel~Bancal}
\affiliation{Group of Applied Physics, University of Geneva, CH-1211 Geneva, Switzerland}
\affiliation{Center for Quantum Technologies, National University of Singapore, 3 Science Drive 2, Singapore 117543}
\author{Yeong-Cherng~Liang}
\affiliation{Group of Applied Physics, University of Geneva, CH-1211 Geneva, Switzerland}
\author{Martin~Hofmann}
\affiliation{Naturwissenschaftlich-Technische Fakult\"at, Universit\"at Siegen, Walter-Flex-Str.~3, D-57068 Siegen, Germany}
\author{Otfried~G\"uhne}
\affiliation{Naturwissenschaftlich-Technische Fakult\"at, Universit\"at Siegen, Walter-Flex-Str.~3, D-57068 Siegen, Germany}

\begin{abstract}
We present a general method to quantify both bipartite and multipartite entanglement in a device-independent manner, meaning that we put a lower bound on the amount of entanglement present in a system based on observed data only but independently of any quantum description of the employed devices. Some of the bounds we obtain, such as for the Clauser-Horne-Shimony-Holt Bell inequality or the Svetlichny inequality, are shown to be tight. Besides, device-independent entanglement quantification can serve as a basis for numerous tasks. We show in particular that our method provides a rigorous way to construct dimension witnesses, gives new insights into the question whether bound entangled states can violate a Bell inequality, and can be used to construct device independent entanglement witnesses involving an arbitrary number of parties.
\end{abstract}


\maketitle

\textit{Introduction.}---Entanglement, undoubtfully the most precious resource of quantum mechanics, has been routinely quantified in many experiments. However, such entanglement statements are generally only valid when a precise quantum description of the employed equipment is available~\cite{Rosset12}. In many contexts, such a quantum model is not available, in particular for complex biological or condensed matter systems, where one still disputes about the underlying quantum processes or is unsure about the appropriate description of measurements~\cite{tiersch12a,cramer11a}. In this case, one can still try to quantify entanglement \emph{exclusively} from the observed classical measurement data, thus independent of any quantum functionality of the interested system. While this may seem impossible at first sight, such methodology is precisely the working principle behind the emergent field of \emph{device-independent} quantum information processing, which started in quantum key distribution~\cite{ekert91a,acin06a} and device testing~\cite{mayers04a,magniez05a}. However, while it is long known that Bell inequality violations~\cite{bell64} verify entanglement~\cite{werner89a}, no precise bound on the amount of entanglement is known in the device-independent setting, presumably because non-locality and entanglement are different resources~\cite{method07a}. Even with a qubit assumption, quantification has so far only been achieved for the simplest experimental scenario~\cite{Verstraete02,liang11a}.

In this paper we present a general framework for various device-independent tasks, notably the quantification of bi-  and multipartite entanglement using solely the observed classical data. Incidentally, this provides further results on seemingly unrelated questions in quantum information: First it certifies a necessary minimal dimension of the underlying quantum system and thus provides a rigorous and systematic construction of dimension witnesses~\cite{brunner08a}. Second, using the negativity~\cite{vidal02a} as our primary entanglement measure, we obtain new results for the long-standing Peres conjecture~\cite{peresconjecture}, which states that no bound entangled state can violate a Bell inequality. We show that a Bell violation of any known bipartite bound entangled state, or more precisely, any entangled state with a positive partial transpose (PPT), can at most be very small, if not vanishing, for the simplest classes of Bell inequalities, thus providing circumstantial evidence in favor of this conjecture in the bipartite case. Finally, in the multipartite case our framework additionally facilitates---without resorting to the detection of genuine multipartite nonlocality~\cite{bancal11a}---the construction of device-independent entanglement witnesses (DIEW) for genuine multipartite entanglement~\cite{bancal11a,pal11a,Bancal12a}.

\textit{Problem definition.}---Let us start by considering a bipartite Bell-type experiment where each party can employ different measurement settings $x,y$ with respective outcomes $a,b$ that are sampled from the conditional probability distribution $P(a,b|x,y)$. These data have a quantum representation if there exists a quantum state $\rho_{AB}$ and local measurement operators $M_{a|x}$, $M_{b|y}$ such that $P(a,b|x,y)=\tr( \rho_{AB} M_{a|x} \otimes M_{b|y})$. In the device-independent paradigm one tries to draw conclusions about $\rho_{AB}$ directly from $P(a,b|x,y)$ without assuming any knowledge of the performed measurements or of the dimension of the underlying state. In order to do so one needs a characterization at the level of $P(a,b|x,y)$ assuming that $\rho_{AB}$ satisfies certain properties. If $\rho_{AB}$ is only required to be a quantum state, we recover the original question leading to Tsirelson's bounds~\cite{tsirelson80a,wehner06a,navascues07a,navascues08a,doherty08a}. But one can demand $\rho_{AB}$ to fulfill extra constraints, such as being PPT~\cite{peres96a}, or---with our primary goal in mind---that its entanglement is bounded. This characterization task generalizes naturally to the multipartite case, \eg, to describe if the tripartite distribution $P(a,b,c|x,y,z)$ is quantum, biseparable~\cite{duer99a,bancal11a}, originates from a PPT mixture~\cite{jungnitsch11a} or has some bounded amount of entanglement.

Our method is a superset characterization, similar to the converging hierarchy proposed by Navascu\'es-Pironio-Ac\'in (NPA)~\cite{navascues07a,navascues08a,doherty08a}. For instance, in the bipartite case we show that a distribution $P=P(a,b|x,y)$ can only originate from a PPT state if a special matrix $\chi[P,u]$, that linearly depends on $P$ and on some unknowns $u$, satisfies $\chi[P, u]\geq 0$ and $\chi[P,u]^{T_A}\geq 0$. If it is impossible to find such parameters $ u$, then $P$ has no PPT quantum representation. The novel observation which enables us to go beyond NPA is that if one organizes the matrix entries of NPA carefully, the resulting matrix $\chi$ can be interpreted as the result of local maps acting on the underlying quantum state. Then this matrix has a clear bipartite structure.

We emphasize that in quantifying entanglement or in characterizing correlations due to extra properties of the quantum state, we need statements that hold for all possible dimensions, measurements and states with the desired property. However, since any measurement operator corresponds to a projector in higher dimensions, we can assume without loss of generality the projection property, i.e., the relation $M_{a|x}M_{a^\prime|x} = \delta_{aa^\prime} M_{a|x}$ for the operators $M_{a|x}$ on system $A$ for all $x$, $a$ and $a'$. This follows from Naimark's extension~\cite{Peres:1990fk} which preserves any entanglement monotone. Also, we shall simultaneously employ the notations $M_{a|x}$ and $A_i$ for measurement operators on system $A$, likewise for other systems. The set $\left\{ A_i\right\}$ contains the identity operator $A_0=\mathbbm{1}$ and all but one measurement operator $M_{a|x}$ for each setting. Hence one has the aforementioned projection property and an identity relation $A_i A_0 = A_0 A_i = A_i$ for all $i$.

\textit{Technique.}---To solve to desired characterization problem, we employ results obtained in the studies of matrix of moments for continuous variable systems~\cite{shchukin05a,rigas06a,haeseler08a,miranowicz09a,killoran11a} and  in the device-independent analysis~\cite{wehner06a,navascues07a,navascues08a,doherty08a}.

Let us start with the matrix of moments for the bipartite case and consider first the scenario where the state $\rho_{AB}$ and measurement operators $M_{a|x}, M_{b|y}$ are known. To this scenario we associate two completely positive (CP) \emph{local} maps $\Lambda_A,\Lambda_B$ that we apply to the quantum state $\chi[\rho] = \chi[\rho_{AB}]_{\bar A \bar B} = \Lambda_A \otimes \Lambda_B[\rho_{AB}]$. Here $\bar A$ and $\bar B$ denote the respective output spaces. Specifically, consider the local map $\Lambda_A[\rho] = \sum K_n\, \rho\, K_n^\dag$ where the Kraus operators are given by $K_n = \sum_i \ket{i}_{\bar A \: A\!}\bra{n}A_i$, and $\ket{n}_A$, $\ket{i}_{\bar A}$ are orthogonal basis states of $\mathcal{H}_A$ and $\mathcal{H}_{\bar A}$ respectively. Using a similar map for $B$ one obtains 
\begin{equation}
\label{eq:chi}
\chi[\rho] = \sum_{ijkl} \ket{ij}_{\bar A \bar B}\bra{kl} \tr[\rho_{AB} A_k^\dag A_i \otimes B_l^\dag B_j].
\end{equation}
Thus the matrix $\chi[\rho]$ is just a matrix of certain expectation values. Since the local maps can also be defined using higher moments, \eg, by choosing Kraus operators $K_n = \sum_{i_1,\dots, i_{\ell}} \ket{i_1,\dots,i_\ell}_{\bar A \: A\!}\bra{n}A_{i_1}A_{i_2}\dots A_{i_\ell}$, we shall refer to $\chi$ as a moment matrix of level $\ell$ if it contains all $\ell$-fold products of $A_i$. Since both sets $\{A_i\},\{B_j\}$ contain the identity, the trace of the underlying state is a matrix entry that we refer to as $\chi[\rho]_{\tr}=\tr[\rho]$. Finally, note that by the structure of these local maps we have a couple of important relations, e.g.: i) if $\rho \geq 0$ then $\chi[\rho] \geq 0$, ii) if ${\rho^{T_A} \geq 0}$ then $\chi[\rho]^{T_{\bar A}} \geq 0$, and iii) if $\rho$ separable then $\chi[\rho]$ separable. This matrix of moment approach can analogously be defined in the multipartite case.

A device-independent characterization draws conclusion only from the observed correlations, hence, many of the entries of $\chi$ are unknown a priori. However, even without this information the matrix $\chi[\rho]$ has a structure which follows from known relations that hold independently of state and measurements: $1$) $A_i, B_j$ are Hermitian operators, $2$) $A_i,B_j$ satisfies the above mentioned projection property and the identity relation, $3$) certain entries correspond to the observations $P(a,b|x,y)=\tr(\rho_{AB} M_{a|x} \otimes M_{b|y})$. 

Via this partial information we can decompose without loss of generality each matrix of moments $\chi[\rho]$ as
\begin{align}
\nonumber
\chi[\rho] &= \chi[P,u] = \chi^{\rm fix}(P) + \chi^{\rm open}(u)\\
\label{eq:chi_sdp}
&= \sum_{a,b,x,y} P(a,b|x,y) F_{abxy} + \sum_v u_v F_v,
\end{align}
i.e., into one fixed part that linearly depends on the observed data $\chi^{\rm fix}(P)=\sum P(a,b|x,y) F_{abxy}$ and into an orthogonal, open part $\chi^{\rm open}(u)=\sum_v u_v F_v$ which would be known only by the knowledge of state and measurements. Here all operators $F=F^\dag$ are Hermitian. Note that the constraint $\chi[\rho]_{\tr} = \chi[P,u]_{\rm tr} = 1$ is fulfilled automatically if the probabilities $P$ are normalized. We give an example how the relations $1)-3)$ provide the form given by Eq.~(\ref{eq:chi_sdp}) in the appendix.

\textit{Connection with the NPA hierarchy.}---At this point we like to connect the present technique to that of NPA~\cite{navascues07a,navascues08a}, the best known method to characterize quantum correlations. For their method, one can identify a likewise construction $\chi^{\rm NPA}[\rho] = \Lambda[\rho_{AB}]$, but with $\Lambda$ being a \emph{global} CP map which already ensures that if $\rho\geq 0$ then $\chi^{\rm NPA}[\rho] \geq0$. If one uses the operator-sum ansatz $\chi^{\rm NPA}[\rho] = \sum_m L_m \rho\, L_m^\dag$ where $L_m = \sum_s \ket{s}\bra{m}O_s$ with $\ket{m},\ket{s}$ being respective basis states for the bipartite in- and output Hilbert spaces, this leads to $\chi^{\rm NPA}[\rho] = \sum \ket{s}\bra{t} \tr[\rho_{AB} O_t^\dag O_s]$. If this operator set $\{O_s\}$ consists of all $\ell$-fold products of measurement operators, then imposing the constraint $\chi^{\rm NPA}[\rho]\geq 0$ corresponds to the $\ell^{\rm NPA}$-th step in their hierarchy.

Therefore a bipartite moment matrix $\chi$ of level $\ell$ as defined above and a $2\ell$-step $\chi^{\rm NPA}$ only differ in the ordering of the expectation values and in that certain moments of $\chi^{\rm NPA}$ are not included in $\chi$. These similarities are important to relate results about the NPA method $\chi^{\rm NPA}$ to the modified moment matrix $\chi$. However, let us stress that $\chi^{\rm NPA}$ does not generally admit a bipartite structure.

\textit{Applications of technique.}---Given the close connection between the present technique and that of NPA, it is clear that ours can also be used to characterize the set of quantum correlations and hence to compute Tsirelson bounds, \ie, extremal quantum values of a Bell inequality. For instance, for any fixed level $\ell$  and any given Bell expression $I\cdot P=\sum I_{abxy} P(a,b|x,y)$, an upper bound to Tsirelson bound can be obtained by solving $\max \{ I\cdot P\, |\, \chi[\rho]= \chi[P,u] \geq 0\}$ as a semidefinite program~\cite{cobook}. Henceforth, let us focus on the novel applications that stem from the current technique.

In comparison with NPA the advantage of the additional bipartite structure $\chi=\chi_{\bar A \bar B}$ is that one can now easily incorporate further constraints. For instance, one could ask for a similar Tsirelson bound if the underlying state is PPT by including the constraint $\chi[\rho]^{T_{\bar A}} \geq 0$, 
\begin{eqnarray}
\label{eq:ppt_tsirelson}
	\underset{P,u}{\max} && \: I \cdot P \\ \nonumber
	\sthat && \: \chi[\rho]=\chi[P,u] \geq 0, \; \chi[\rho]^{T_{\bar A}} \geq 0,\; \chi[\rho]_{\tr}=1.
\end{eqnarray} 
By this method one obtains an upper bound to the true PPT Tsirelson bound, which converges to the related commutative bound in the limit of large levels $\ell$, see appendix for details.
 
Next, let us show how to estimate the negativity~\cite{vidal02a}, defined via the sum of negative eigenvalues $\lambda_i$ of the partially transposed state as $N[\rho_{AB}] = \sum_{\lambda_i<0} |\lambda_i(\rho_{AB}^{T_A})|$. In the following we employ its variational form which reads as $N[\rho_{AB}]\! = \!\min \{ \tr[\sigma_-]| \rho_{AB}\! = \sigma_+\! -\! \sigma_-, \sigma_\pm^{T_A}\geq0 \}$. Using the properties of the moment matrix,  one can readily optimize over a larger set: The constraint $\rho = \sigma_+ - \sigma_-$ is relaxed by $\chi[\rho] = \chi[\sigma_+] - \chi[\sigma_-]$, while $\sigma_\pm^{T_A} \geq 0$ translates to $\chi[\sigma_\pm]^{T_{\bar A}} \geq 0$. 
If one observes a certain violation of a Bell inequality $I \cdot P = v$, a lower bound on the negativity of $\rho_{AB}$ compatible with this observation is given by 
\begin{eqnarray}
\label{eq:dvi_negativity}
\underset{P,u,P_\pm,u_\pm}{\min} && \!\!\!\!\!\!\: \chi[\sigma_-]_{\tr}\\
\nonumber
\sthat \:\:\:\: && \!\!\!\!\!\!\!\: \chi[\rho] = \chi[P,u] = \chi[\sigma_+] - \chi[\sigma_-] \geq 0,\: \chi[\rho]_{\rm tr} = 1,\\
\nonumber
&& \!\!\!\!\!\!\!\: \chi[\sigma_\pm]^{T_{\bar A}}  = \chi[P_\pm,u_\pm]^{T_{\bar A}} \geq 0,\:I \cdot P = v.
\end{eqnarray}

Furthermore, since the negativity of any $\mathbb{C}^{d}\otimes \mathbb{C}^{D}$ state is at most $N_{\rm max}^d=(d-1)/2$ (for $d\leq D$), a lower bound on the negativity certifies also a minimal state space dimension. The bound of a dimension witness~\cite{brunner08a}, i.e., the maximal value of a Bell inequality for states with minimal local dimension upper bounded by $d$, can be constructed by an optimization  analogous to Eq.~\eqref{eq:dvi_negativity} but with the expression $I\cdot P$ now appearing in the objective function, while the dimension restriction is enforced by the constraint $\chi[\sigma_-]_{\rm tr} \leq N_{\rm max}^d$. 

At this point we like to stress that these optimization problems  admit a natural generalization to the multipartite scenario using PPT mixtures (which include biseparable states) and the genuine negativity as a measure for genuine multiparticle entanglement~\cite{jungnitsch11a}. Further details and the explicit programs are given in the appendix.

\textit{Example I: CHSH}---Let us start with the Clauser-Horne-Shimony-Holt (CHSH) inequality~\cite{clauser69a}, where each party has two possible settings $x,y\in\{1,2\}$ yielding binary outcomes $a,b$. Using correlation terms $\braket{X_xY_y}=P(a=b|x,y)-P(a\not=b|x,y)$, the inequality 
$I_{\rm CHSH}=\braket{X_1Y_1} + \braket{X_1Y_2} + \braket{X_2Y_1} - \braket{X_2Y_2} \leq 2
$
holds for any local hidden-variable model (LHV), while quantum mechanics allows a maximum of $I_{\rm CHSH}^{\rm max} = 2\sqrt{2}$. Since every separable state fulfills the LHV bound~\cite{werner89a}, any violation $I_{\rm CHSH} > 2$ signals entanglement of the underlying quantum state $\rho_{AB}$. By solving Eq.~\eqref{eq:dvi_negativity} we can now provide a quantitative statement in terms of the minimal negativity that the underlying state $\rho_{AB}$ must possess. Specifically, the numerical result leads to the sharp bound \begin{equation}
N[\rho_{AB} \left| I_{\rm CHSH}=v\right.]\ge (v-2)/(4\sqrt{2}-4).
\end{equation}
The resulting plot and a more detailed discussion, also about the other examples, can be found in the appendix. Note that this recovers the known result that PPT states must necessarily satisfy the CHSH inequality~\cite{werner00a}.

\begin{figure}[Ht]\center
\includegraphics[angle=-90,scale=0.35]{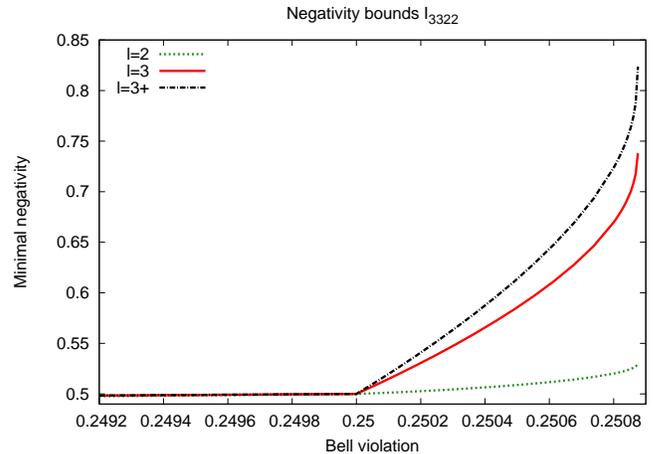} 
\caption{Negativity bounds for violations close to the maximum of the Bell inequality $I_{3322}$~\cite{sliwa03a,collins04a} obtained by solving Eq.~(\ref{eq:dvi_negativity}) for different levels of the moment matrix. Note that violations with $v>0.25$ require a negativity of $N[\rho_{AB}]>1/2$ and thus at least a two-qutrit state.}\label{fig:bell2negativity}
\end{figure}

\textit{Example II: Dimension witness}---As a second example, we consider the Bell inequality $I_{3322}\leq 0$ ~\cite{sliwa03a,collins04a} where each party can perform three possible dichotomic measurements as indicated by the subscripts. For a violation of $0\le v\le 0.25$, the numerical solution of Eq.~\eqref{eq:dvi_negativity} gives $N[\rho_{AB}\left| I_{3322}=v\right.]\geq 2v$ and a two-qubit Bell state can indeed reach a violation of $I_{3322}=0.25$~\cite{collins04a}. However the maximum possible quantum violation is given by $I_{3322}^{\rm max} \lesssim 0.25088$ and there exist infinite-dimensional states which can asymptotically reach this value~\cite{pal10a}. From Fig.~\ref{fig:bell2negativity}, we see more closely that if $I_{3322}>0.25$ the negativity bound satisfies $N[\rho_{AB}]>1/2$, which is achievable only with local Hilbert space dimension $d\ge 3$. Hence,  $I_{3322}\le 0.25$ serves as a dimension witness for qutrits. 
In a similar way we investigated the very first Bell inequality used as a dimension witness~\cite{brunner08a}, namely, $I_{2233}\leq 0$~\cite{collins04a,kaslikowski02,collins02a}, and confirm that violations larger than $v=1/\sqrt{2}-1/2 \approx 0.2071$ require at least qutrits---this certifies the heuristic qubit bound of $I_{2233}$~\cite{brunner08a}. 

\textit{Example III: PPT Tsirelson bound}---As a third example of the application of our techniques, we have computed upper bound on the PPT Tsirelson bound for the above Bell inequalities and $175$ facet-defining Bell inequalities involving four dichotomic measurement settings per party~\cite{Brunner20083162,Pal09,TVpv}. Interestingly, our results show that for the majority of these inequalities, the maximal quantum violation allowed by all PPT entangled states is vanishing within numerical precision, hence unable to provide a counterexample to the bipartite Peres conjecture, cf.  Tab.~\ref{tab:PPTtsireslon} and the appendix for more details.

\begin{table}[h!bp]\center
\begin{tabular}{cccccc}
Bell inequality &\phantom{sp}& PPT Tsirelson bound  &\phantom{sp}& $\ell$ & Matrix size \\
\hline \hline
\T $I_{\rm chsh} \leq 2$~\cite{clauser69a} & &$2$ & &1 &$3^2\times 3^2$\\
\T $I_{3322} \leq 0$~\cite{collins04a} & &$0$ & &2 & $10^2\times10^2$\\
\T $I_{2233} \leq 0$~\cite{collins04a} & &$\leq 1.3559\times 10^{-3}$ & &$3+$ & $45^2\times45^2$\\ 
\T $A_{6} \leq 0 $~\cite{Brunner20083162} & &$\leq  7.6754\times 10^{-6}$ & &$2+$ & $31^2\times31^2$\\
\T\B $I_{4422}^3 \leq 0 $~\cite{Brunner20083162} & &$\leq  2.8531\times10^{-4} $ & &$2+$ & $27^2\times27^2$\\
\hline
\T\B $I_{S5} \leq 3$~\cite{sliwa03a} & & $\leq  3.0187$& &$3$ & $7^3\times 7^3$\\
\hline \hline
\end{tabular}
\caption{Upper bounds on the maximal possible violation by PPT states for different Bell inequalities computed via Eq.~\eqref{eq:ppt_tsirelson}. All the other four dichotomic setting Bell inequalities investigated have a PPT Tsirelson bound which is already within the respective LHV bound by less than $10^{-6}$. Further specifications: $\ell$ labels levels of the matrix of moment, ``Matrix size'' refers to the dimension of the moment matrix. The last inequality corresponds to the tripartite case for states which are PPT for all bipartitions. The precision is at least $10^{-7}$.}\label{tab:PPTtsireslon}
\end{table}

\textit{Multipartite case}---We also considered examples involving more than two parties, where one is typically interested to verify genuine multipartite entanglement. This strongest form of multiparticle entanglement can be detected from observed correlations alone by violating a DIEW~\cite{bancal11a}. For device-independent entanglement quantification, we investigated---by a method analogous to the bipartite case---the minimal amount of genuine negativity~\cite{jungnitsch11a} needed to violate the DIEWs $I_{32}$ and $I_{33}$, where each party has respectively two or three dichotomic measurements~\cite{bancal11a,Barreiro}. Since $I_{32}$ is the Svetlichny inequality~\cite{svetlichny87}, its violation also demonstrates genuine multipartite nonlocality. From the bounds we computed, again tight for the Svetlichny case (see appendix), we can also obtain information about the type of entanglement responsible for given violations, in similar spirit to Ref.~\cite{brunner11a}. For instance, since the genuine negativity of any state of the three-qubit $W$-class~\cite{acin01a} is bounded by $\sqrt{2}/3$, one verifies that violations close to the maximum of these DIEWs can never be achieved by such type of entanglement. Moreover, our bounds show that these DIEWs can never be violated by states which are PPT mixtures~\cite{jungnitsch11a}. Using similar arguments as presented in Ref.~\cite{Bancal12a}, this result can even be extended to the $n$-partite witnesses $I_{n2}$ and $I_{n3}$. This suggests that, apart from a quantification, the generalization of PPT Tsirelson bounds to the multipartite case provides a tractable way to approximate the set of biseparable quantum correlations in the presence of more than three parties~\cite{bancal11a}.  Indeed, this approximation not only works well for the two families of $n$-partite DIEWs $I_{n2}$, $I_{n3}$, but also for a large number of symmetric 4-partite DIEWs involving two dichotomic measurements~\cite{liang13}.

Finally, there are also other questions for the multipartite case. At last we computed the maximal violation of the tripartite Bell inequality $I_{{S5}}\leq 3$~\cite{sliwa03a} for states which are PPT for all bipartitions. We find that it is bounded by $3.0187$, which shows that the example of Ref.~\cite{vertesi12a}, optimally violates the tripartite Peres conjecture via this inequality, cf. Tab.~\ref{tab:PPTtsireslon}.

\textit{Conclusion.}---We have presented a versatile tool to quantify entanglement in the bi- and multipartite case directly from the observed  measurement results, thus irrespective of any quantum functionality of the employed devices. This framework offers great practical benefit in experiments since its statements are robust against any kind of systematic errors in the assumed quantum model and involves minimal assumptions. Moreover such a quantification provides additional applications: It yields information about the underlying state space dimension or the type of entanglement involved in the multipartite case. Furthermore, our tool allows for a systematic investigation into the long-standing Peres conjecture, and the computation of device independent entanglement witness for genuine multipartite entanglement.

For future work, we believe that our method can be extended to bound, in a device-independent manner, other entanglement measures. Clearly, it will also be interesting to investigate how our technique can be used in conjunction with other separability criteria, or applied in the closely-related steering~\cite{wiseman07} (with the partial information step only applied to one-side) or sequential measurement scenarios~\cite{budroni13a}.

\begin{acknowledgments}
This work greatly profited from discussions with R.~Augusiak and J.~I. de Vicente at the Centre of Sciences in Benasque. We would also like to thank M.~Navascu{\'e}s, K.~F.~P\'al, V.~B. Scholz and T.~V\'ertesi for stimulating discussions about topic and technicalities. This work has been supported by the EU (Marie Curie CIG 293993/ENFOQI), the BMBF (Chist-Era project QUASAR), the Chist-Era project DIQIP, the Swiss NCCR-QSIT, the National Research Foundation and the Ministry of Education of Singapore.
\end{acknowledgments}

\appendix

\section{Example of constraints on the matrix of moments}

As an example of how the relations $1)-3)$ from the main text lead to the general structure of $\chi[\rho]$ as given by Eq.~($2$), we consider the case where the map $\Lambda_A$ is applied to a single system $\rho_A$, \ie, $\chi[\rho] = \sum_{ij} \ket{i}\bra{j}\tr[\rho_A A_j^\dag A_i]$. Suppose that the first setting $x$ has three outcomes $\{0,1,2\}$ whereas the second one, here labeled as $x^\prime$, has only two outcomes $\{0,1\}$. Using $\{A_i\} = \{ \mathbbm{1}, M_{0|x},M_{1|x},M_{0|x^\prime} \}$ one gets 
\begin{equation}
\!\!\chi[\rho]\!=\!\! \left[\!\! \begin{array}{cccc} 
\mathbbm{1}^\dag \mathbbm{1} & \mathbbm{1}^\dag M_{0|x} & \mathbbm{1}^\dag M_{1|x} & \mathbbm{1}^\dag M_{0|x^\prime} \\ 
M_{0|x}^\dag \mathbbm{1} & M_{0|x}^\dag M_{0|x} & M_{0|x}^\dag M_{1|x} & M_{0|x}^\dag M_{0|x^\prime} \\ 
M_{1|x}^\dag \mathbbm{1} & M_{1|x}^\dag M_{0|x} & M_{1|x}^\dag M_{1|x} & M_{1|x}^\dag M_{0|x^\prime}\\
M_{0|x^\prime}^\dag \mathbbm{1} & M_{0|x^\prime}^\dag M_{0|x} & M_{0|x^\prime}^\dag M_{1|x} & M_{0|x^\prime}^\dag M_{0|x^\prime} \end{array} \!\!\right]_{\rho_A}\!\!\!\!\!
\end{equation}
where $\rho_A$ indicates that we still must take the expectation values. Via the listed properties $1)-3)$ one obtains
\begin{equation}
\chi[\rho]=\left[ \begin{array}{cccc} 
\tr(\rho) & P(0|x) &P(1|x) & P(0|x^\prime)\\ 
P(0|x) & P(0|x) &  0 & u_1 + \mathbbm{i} u_2\\ 
P(1|x) & 0 & P(1|x) & u_3 + \mathbbm{i} u_4\\
P(0|x^\prime) & u_1 - \mathbbm{i} u_2 & u_3 - \mathbbm{i} u_4 & P(0|x^\prime) \end{array} \right].
\end{equation}
For the diagonal entries one employs $\tr[\rho_A M_{0|x} M_{0|x}] = \tr[\rho_A M_{0|x}] = P(0|x)$, while the zero entries occur by the projection identity $M_{0|x} M_{1|x}=0$. Since the expectation value of $M_{0|x} M_{0|x^\prime}$ is not directly accessible we can only set it equal to a general complex entry $u_1 +\mathbbm{i} u_2$ using real coefficients $u$.  

Note that whenever the underlying state $\rho$ is normalized, the first term of $\chi[\rho]$ is fixed to be $\tr(\rho)=1$. In the case where the underlying operator does not need to be normalized, however, as it happens for instance with $\sigma_{\pm}$ in the negativity estimation by Eq.~($3$), this entry is not fixed a priori. It is thus given by an unknown variable $u_0\in\mathbbm{R}$.

\section{The multipartite scenario}

In this section we define more precisely moment matrices for an arbitrary number of parties, and present multipartite optimization problems that compute the Tsirelson bounds for PPT mixtures and the device-independent quantification of entanglement in terms of genuine negativity.

In a multipartite scenario one can ask, in analogy with Eq.~($3$), for a bound on the observed correlations if the underlying state is a PPT mixture~\cite{jungnitsch11a}. A $n$-partite state $\rho$ is said to be a PPT mixture if it can be written as a convex combination $\rho = \sum_m p_m \rho_m$ of states $\rho_m$ which are PPT with respect to different bipartitions $m \subset \{1,\ldots,n\}$ of the $n$ subsystems. Since any state that is separable with respect to a chosen bipartition $m$ is also PPT according to this splitting, the set of biseparable states is included in the set of PPT mixtures. Hence if one verifies that a given state is not a PPT mixture, one automatically certifies that it is not biseparable, and thus, by definition, genuine multipartite entangled.

The definition of moment matrices as given in the main text naturally extend to this $n$-party scenario by applying a local CP map $\Lambda_s$ to each subsystem $s=1,\ldots,n$, \ie, $\rho \mapsto \chi[\rho] = \otimes_{s=1}^n \Lambda_s [\rho]$. By the local structure of this transformation one obtains: For any PPT mixture $\rho$ the resulting matrix of moments can be decomposed as $\chi[\rho]=\sum_m p_m \chi[\rho_m]$ with $\chi[\rho_m] \geq 0$ and $\chi[\rho_m]^{T_{\bar m}} \geq 0$, with $\bar m$ referring to the bipartition on the output spaces. Similar properties as given by $1)-3)$ in the main text, constrain the general structure of $\chi[\rho]$ to be $\chi[P,u]$ and one readily obtains a superset approximation for correlations that can be attained via PPT mixtures. Thus if one is interested in the optimal values of a linear expression like $I \cdot P$, where $P$ is generated by a PPT mixture, one obtains a bound by solving
\begin{eqnarray}
\label{eq:ppt_genuine_tsirelson}
\max && \: I \cdot P\\
\nonumber
\sthat, &&\: \chi[\rho] = \chi[P,u] = \sum \chi[p_m \rho_m] = \sum \chi[P_m,u_m],\\
\nonumber
&&\: \chi[\rho]_{\rm tr} = 1,\: \chi[p_m \rho_m] \geq 0, \: \chi[p_m \rho_m]^{T_{\bar m}}  \geq 0\quad\forall\, m.
\end{eqnarray}
Note that in this formulation we included already the probabilities $p_m$ into the matrix $\chi[P_m,u_m]$ such that $\chi[P_m,u_m]_{\tr}=p_m$ in this case. We show later in the appendix that even this multipartite extension converges in the limit of an infinite number of moments in $\chi[\rho]$.

For the multipartite equivalent of Eq.~($4$), let us first remind that the genuine negativity~\cite{jungnitsch11a} is a computable measure of genuine multipartite entanglement which reduces to the negativity in the bipartite case. It is given by $N_G[\rho] = \min_{\{p_m,\rho_m\}} \sum_m p_m N_m[\rho_{m}]$ where $N_m$ denotes the negativity with respect to bipartition $m$ and the minimization runs over all possible valid state decompositions of the density operator $\rho=\sum_m p_m \rho_m$~\footnote{The dual of the given measure, \ie, the way how it is defined in Ref.~\cite{jungnitsch11a}, would be $N_G(\rho) = - \min \{ \tr[\rho W] | \forall m: W=P_m+Q_m^{T_m},\mathbbm{1}\geq Q_m \geq 0, P_m \geq 0\}$. Hence one sees a missing constraint $\mathbbm{1}\geq P_m$ as compared to its original definition, but this does not alter the properties of this measure, \ie, zero for biseparable states, full LOCC monotone, convexity and invariance under local unitaries.}. In analogy to Eq.~($4$), if one observes a value of $I \cdot P = v$, a lower bound on the genuine negativity compatible with this violation is given by 
\begin{eqnarray}
\label{eq:dvi_genuine_negativity}
\min && \: \!\sum \chi[\sigma_m^{-}]_{\tr}\\
\nonumber
\sthat && \: \!\chi[\rho] = \sum \chi[p_m\rho_m],\: I\cdot P = v,\: \chi[\rho]_{\tr} = 1,\\
\nonumber
&& \: \!\chi[p_m\rho_m] = \chi[\sigma_m^+] - \chi[\sigma_m^-] \geq 0, \: \chi[\sigma_m^\pm]^{T_{\bar m}}  \geq 0\;\; \forall m.
\end{eqnarray}

\section{Additional information on the presented examples}

In this part we present some additional information on the examples mentioned in the main text and their respective negativity bounds. 

First note that any explicit quantum state $\rho_0$ which attains a certain value $v$ of a Bell inequality $I\cdot P=v$ provides an upper bound on the minimal negativity compatible with this violation, i.e., $N_{\rm min}[\rho\,|\,I\cdot P=v]=\min N[\rho\,|\,I\cdot P=v]\leq N[\rho_0]$. Furthermore, since the negativity is convex and invariant by adding local auxiliary states, the minimal negativity is a convex function in the amount of violation $v$, that we shall denote by $f(v) = N_{\rm min}[\rho\,|\,I \cdot P=v]$. If one uses the moment matrices $\chi[\rho]$ with increasing levels $\ell$ one obtains lower bounds $f_{\ell}(v) \leq f(v)$ with increasing accuracy. However, if one finds in the bound $f_{\ell}(v)$ an interval $v \in [a,b]$ such that $f_{\ell}(v)$ is linear and where the endpoints $v=a, v=b$ are known to be attainable by explicit quantum states then $f_{\ell}(v)=f(v)$ is a tight bound of the minimal negativity for $v\in [a,b]$.

\begin{figure}[h!t]
\includegraphics[angle=-90,scale=0.35]{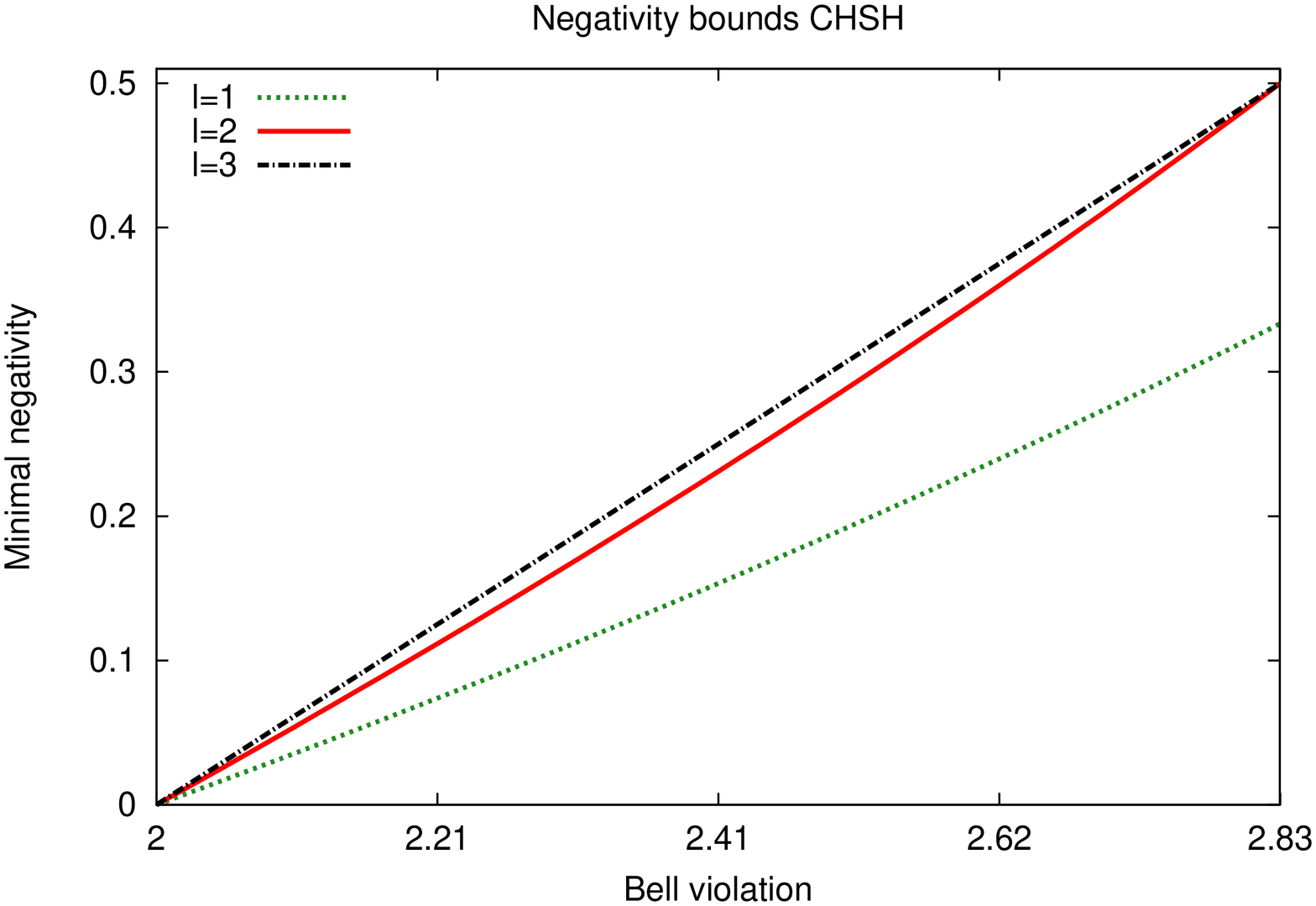} \\
\includegraphics[angle=-90,scale=0.35]{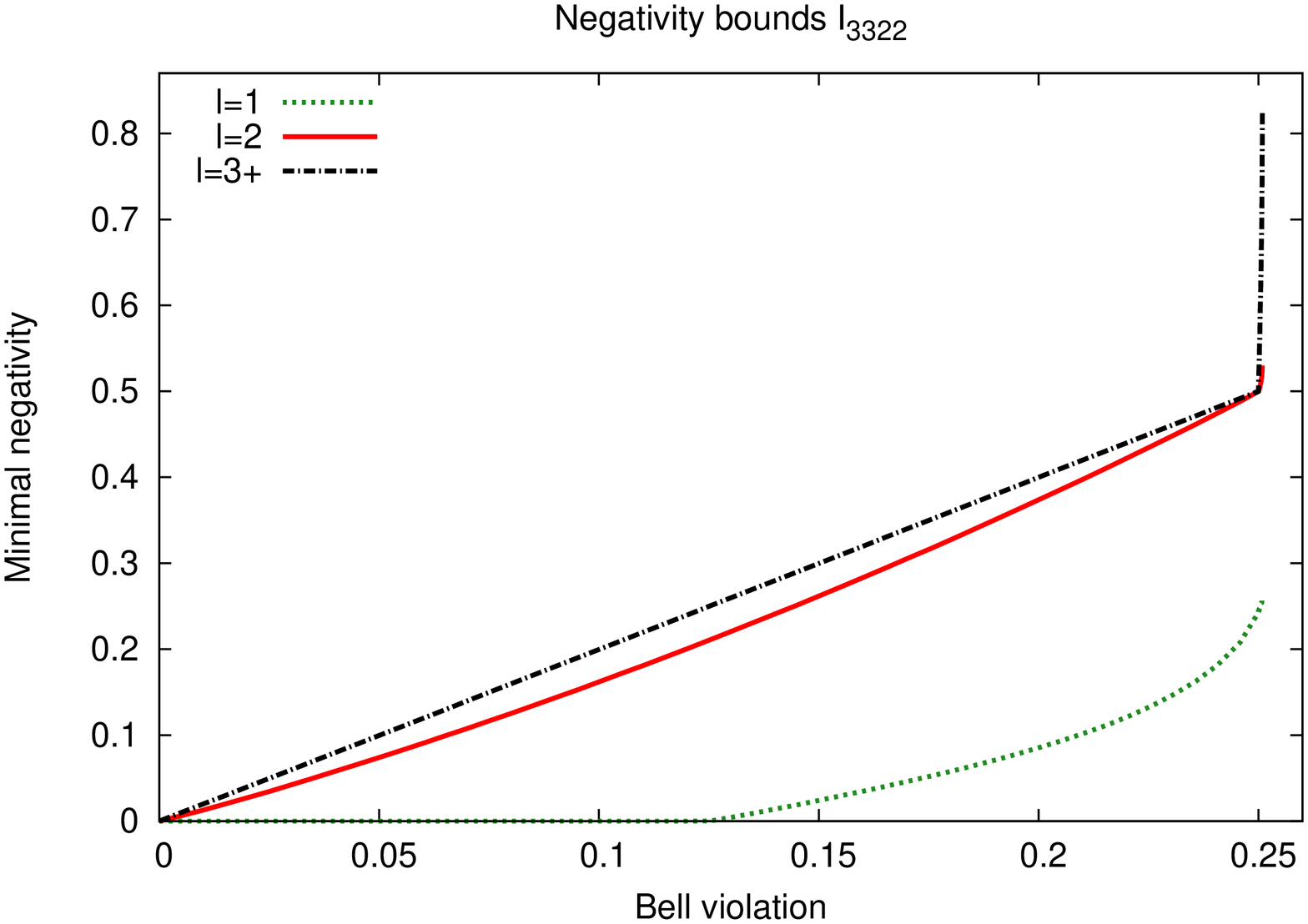} \\
\includegraphics[angle=-90,scale=0.35]{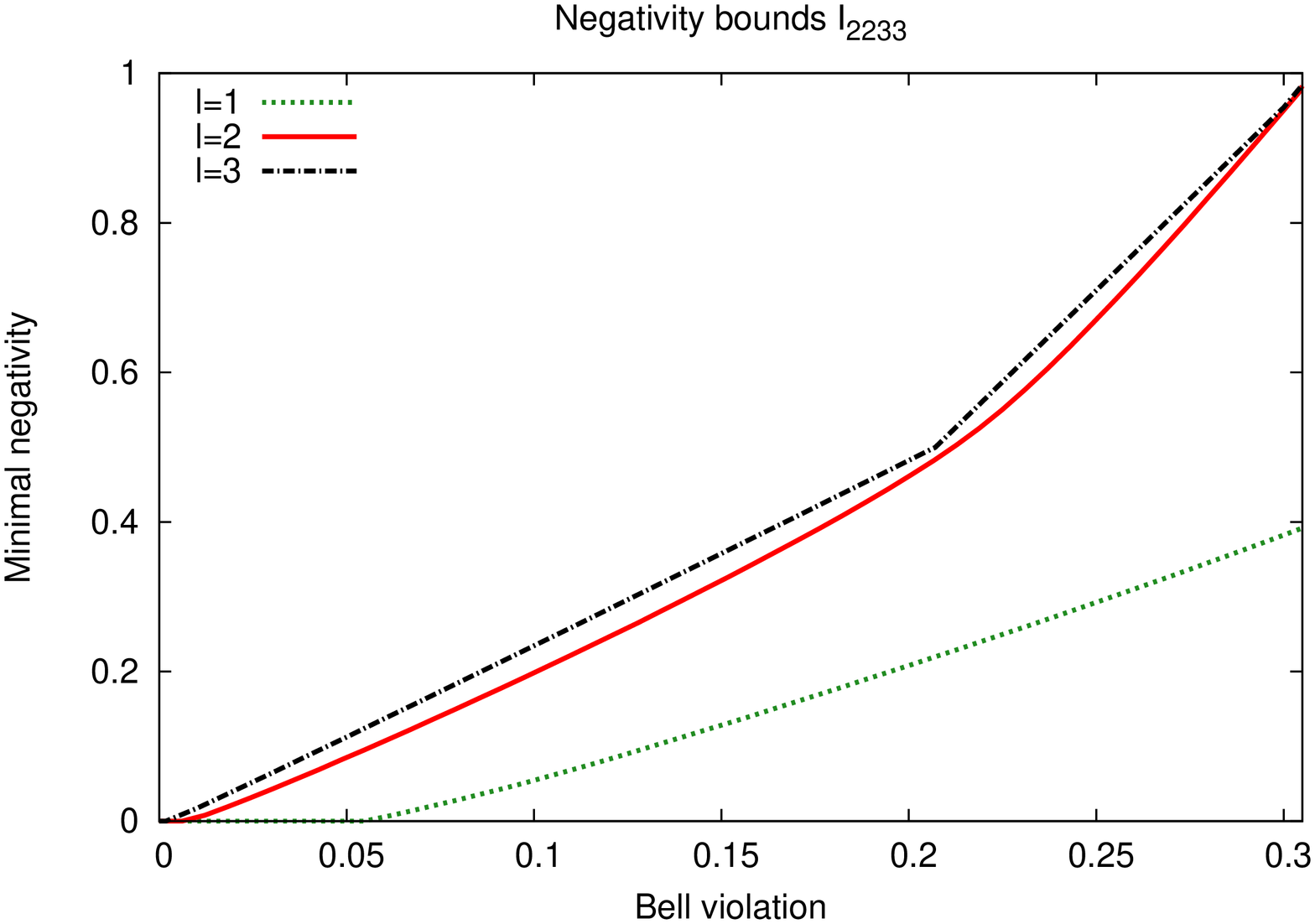}
\caption{Negativity bounds for the CHSH, $I_{3322}$ and $I_{2233}$ Bell inequality as given by Eq.~($4$) for different levels of the moment matrix. The last two inequalities certify a local state space dimension $d\geq 3$ if the negativity exceed $1/2$.}\label{fig:bell2negativity2}
\end{figure}

As we see in more detail in Fig.~\ref{fig:bell2negativity2}, for both, the CHSH inequality and the $I_{3322}$ inequality for a value between $[0,1/4]$, the negativity bounds that we obtained by solving Eq.~($4$) correspond to straight lines for the level $\ell=3$. Specifically, for the CHSH inequality, this negativity bound is a straight line joining the points $(I_{\rm chsh}=2,N[\rho]=0)$ and $(I_{\rm chsh}=2\sqrt{2},N[\rho]=1/2)$ with largest numerical deviation $\approx 7\times10^{-7}$ among all the computed instances. Likewise, for the $I_{3322}$ inequality, it is a linear bound connecting the origin and the coordinate $(I_{3322}=1/4,N[\rho]=1/2)$ with largest numerical deviation $\approx 8\times 10^{-7}$. Since both endpoints $I_{\rm CHSH}=2\sqrt{2}$ and $I_{3322}=1/4$ can be achieved with a maximally entangled two-qubit state, having negativity $N[\rho]=1/2$, we thus arrive at the sharp negativity bounds presented in the main text. 

As we see in Fig~\ref{fig:bell2negativity2}, or more detailed in Fig.~1 from the main text, if $I_{3322}>0.25$ then the negativity bound satisfies $N[\rho_{AB}]>1/2$. Because the negativity of any $\mathbb{C}^2\otimes \mathbb{C}^D$ state is upper bounded by $1/2$ for all $D\ge 2$, this certifies that such violations are achievable only with both local Hilbert space dimension $d\ge 3$, or, in other words, that $I_{3322}\le 0.25$ serves as a dimension witness~\cite{brunner08a} for qutrits.

This insight inspired us to also investigate the very first Bell inequality used to introduce a dimension witness~\cite{brunner08a}, more precisely, $I_{2233}\leq 0$~\cite{collins04a,kaslikowski02,collins02a}, where each party can choose between two $3$-valued outcome measurements. The corresponding negativity bounds are shown in the last plot of Fig.~\ref{fig:bell2negativity2}, and are again tight for the highest computed level, $\ell =3$, by the same arguments as before. Similar to $I_{3322}$ we observe that the negativity crosses $1/2$ at a violation of $v=1/\sqrt{2}-1/2 \approx 0.2071$, therefore a larger quantum violation shows that the underlying state is again at least 3-dimensional. For this inequality it is furthermore interesting that for the maximal violation at $I^{\rm max}_{2233}=(\sqrt{11/3}-1)/3$, our numerical optimization gives a negativity bound that differs from that of the optimal, non-maximally entangled state~\cite{acin02a} $\psi_{\lambda}\propto \ket{00} + \lambda \ket{11} + \ket{22}$ by less than $5\times10^{-6}$.

\begin{figure}[h!t]
\includegraphics[angle=-90,scale=0.35]{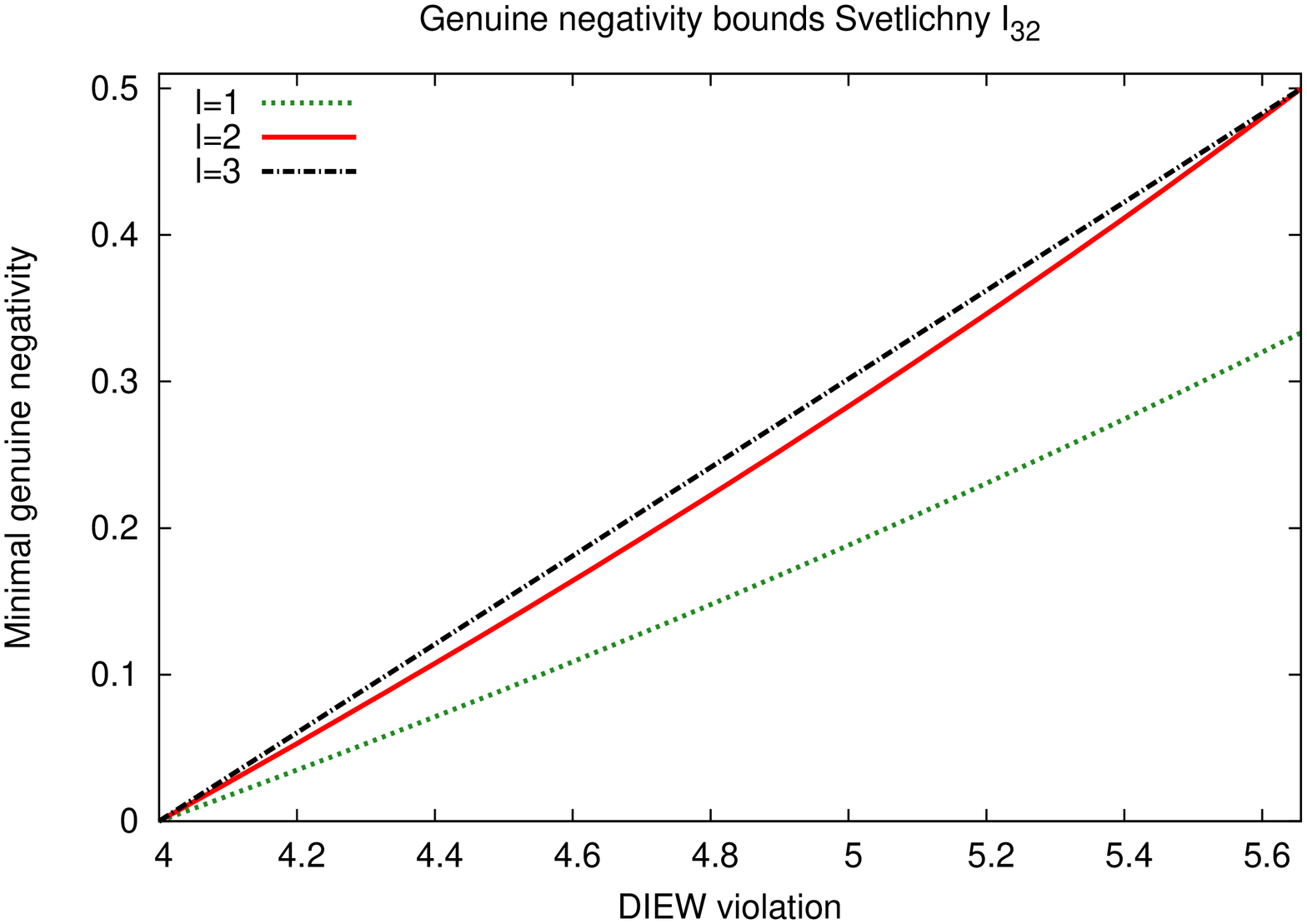}  \\
\includegraphics[angle=-90,scale=0.35]{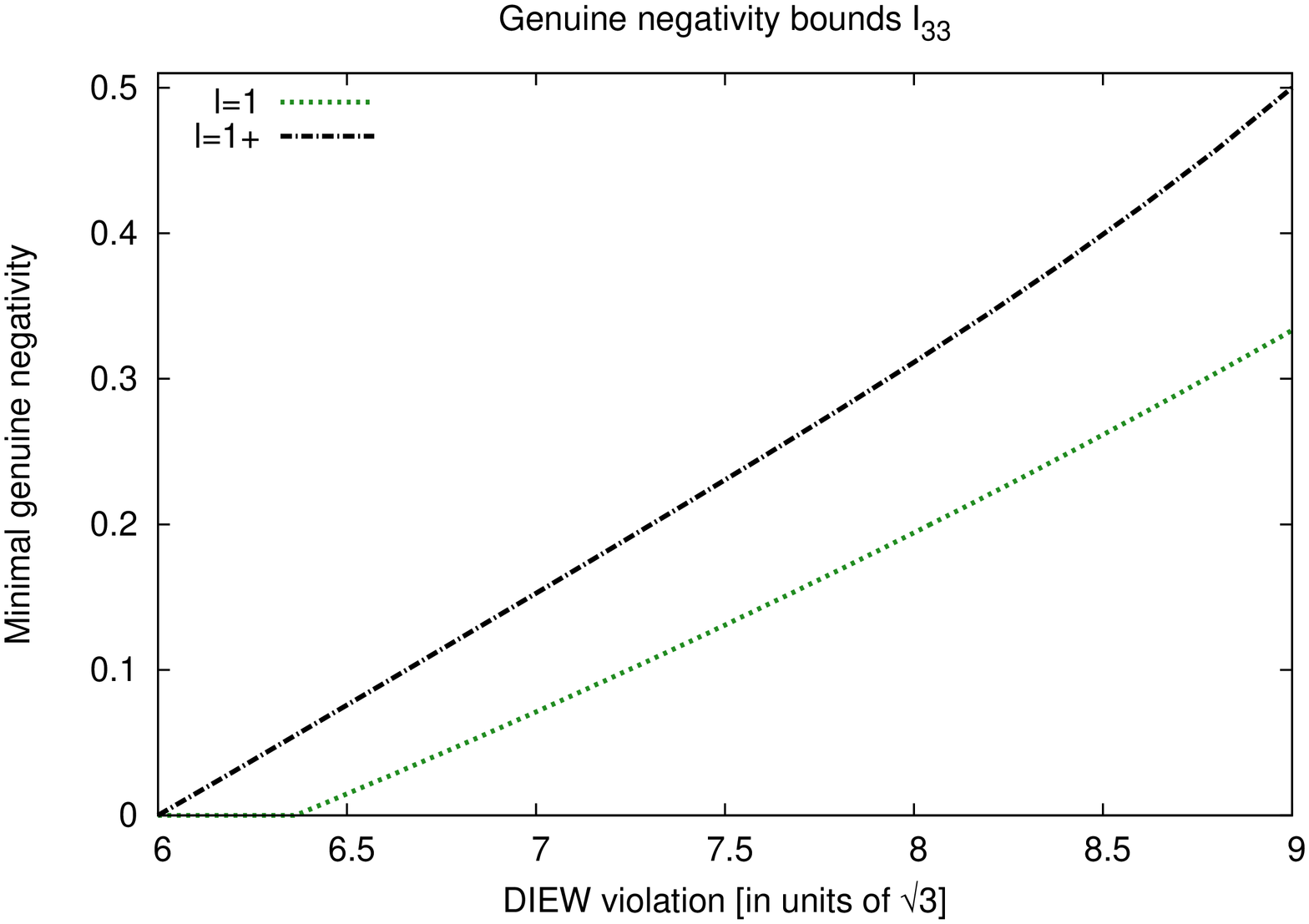}  
\caption{Genuine negativity as a function of the observed violation of the tripartite Svetlichny $I_{32}$ and DIEW $I_{33}$ inequalities. Since any three-qubit state of the $W$-class has a genuine negativity $\lesssim 0.471$,  larger violations certify that the underlying state is not of this SLOCC class.}\label{fig:bell2negativity3}
\end{figure}

In the multipartite case the amount of genuine negativity necessary to achieve different violations of the $I_{32}$ and $I_{33}$ DIEWs were computed according to Eq.~\eqref{eq:dvi_genuine_negativity}. For biseparable states these expressions are bounded by $I_{32}\leq4, I_{33}\leq 6\sqrt{3}$, while the maximal quantum values are $I_{32}^{\rm max}=4\sqrt{2},I^{\rm max}_{33}=9\sqrt{3}$ respectively~\cite{svetlichny87,bancal11a,Barreiro}. The results for different hierarchy levels are shown in Fig.~\ref{fig:bell2negativity3}. Here we refer to partial levels such as level $2+$ to denote moment matrices $\chi$ which were constructed with all terms involving $2$-fold products of local measurement operators, as well as some $3$-fold ones. Similarly to the CHSH and $I_{3322}$ inequalities, the plot for the Svetlichny inequality $I_{32}$ at level $\ell=3$ is a straight line up to numerical precision. Since the value $I_{32}=4\sqrt{2}$ is achievable with a three-qubit GHZ state with genuine negativity $N_G[\rho]=1/2$, we obtain the following tight bound for the minimal genuine negativity compatible with a Svetlichny inequality violation, $N_G[\rho_{ABC}|I_{32}=v] \geq (v-4)/(8\sqrt{2}-8)$.

From Fig.~\ref{fig:bell2negativity3} we also note that no violation of these DIEW is possible with PPT mixtures; hence the bound for biseparable states coincides with the one for PPT mixtures. Since any DIEW with this property can be found by using Eq.~(\ref{eq:ppt_genuine_tsirelson}), and this hierarchy applies to an arbitrary number of parties, this means that our technique can be used to find DIEWs of this kind for an arbitrary number of parties, in contrast to the numerical three-party technique presented in Ref.~\cite{bancal11a}. Furthermore, as pointed out in the main text, since any three qubit state of the W-class~\cite{duer99a} satisfies~\footnote{Since $N_G$ is convex, this optimization can be performed over pure three-qubit states $\ket{\psi_W}$ of the $W$-class. Note that the genuine negativity of a pure state is just the minimum bipartite negativity for all bipartitions $m$. Because $N_G$ is furthermore invariant under local basis changes one can employ the $LU$-equivalent standard form of a pure three-qubit state of the $W$-class~\cite{acin01a}, which leaves a straightforward optimization. The bound is saturated by the $W$-state.} $N_G(\rho_W) \leq \sqrt{2}/3 \approx 0.471$ , our result shows that violations $v > 5.563$ and $v > 15.36$ of the $I_{32}$ and $I_{33}$ inequalities can never be obtained by any such three-qubit states. Therefore one gains even further information from the achieved bound about the underlying type of entanglement, in similar spirit to Ref.~\cite{brunner11a}.

At last let us comment on the presented PPT Tsirelson bounds from the main text, cf. Tab.~I. All bipartite cases are computed directly from Eq.~($3$) and only the last entry corresponds to the multipartite scenario known to provide a counterexample of the multipartite Peres conjecture~\cite{vertesi12a}. For this last Tsirelson bound one optimized the inequality $I_{S5}$ with respect to tripartite states that are PPT for all bipartitions. This optimization problem is like Eq.~($3$) with the tripartite moment matrix and a PPT constraint $\chi[\rho]^{T_m}\geq 0$ for each bipartition. 

Via this numerical investigation on the bipartite case we hoped to find a counterexample to the bipartite Peres conjecture; a PPT state which could violate a Bell inequality. This perspective was triggered by the NPA hierarchy (or using also the modified moment matrix) to compute standard Tsirelson bounds. Although this method is only guaranteed to be complete in the limit of an infinite number of moments, it is important to stress that there are many known instances where one could stop the hierarchy already earlier, since one has already reached the true Tsirelson bound. This is certified by a special rank property of the solution~\cite{navascues08a} and means that the bound does not improve further even if one includes higher moments. However, in all our considered non-trivial PPT examples the respective bounds sharpened if we considered higher levels.

\section{Statement of real states and measurements}

In this section we show that the underlying quantum state and measurement operators can generally be assumed to be real when considering the device-independent quantification of entanglement in terms of (genuine) negativity, i.e., there exists an equivalent real construction having the same (or less) amount of (genuine) negativity. This extends the result of Ref.~\cite{mckague09a} which already proves that probability distributions arising from quantum theory in a Bell-type experiment can always be reproduced using only real states and real measurement operators.

This real property helps in the numerical implementation of the bi- or multipartite programs, Eqs.~($3$),($4$) and Eqs.~\eqref{eq:ppt_genuine_tsirelson},~\eqref{eq:dvi_genuine_negativity} respectively, since it provides a notable parameter reduction in the optimization problems. This reduction becomes even greater in the presence of additional symmetries. We shortly comment on this parameter reduction at the end of this section.

\begin{proposition}\label{Proposition:RealStuff}
Any $n$-partite probability distribution $P$ having a quantum representation with respect to density matrix $\rho$ and measurement operators $A^1_i,A^2_j,\ldots,A^n_k$ (with the superscript labeling the party), also has a real quantum representation, \ie, a representation in terms of a real-valued quantum state $\fixwidehat \rho=\fixwidehat \rho^T$ and real-valued measurement operators $\fixwidehat A_i^1 = \fixwidehat A_i^{1T}, \fixwidehat A_j^2 = \fixwidehat A_j^{2T}, \ldots,\fixwidehat A_k^n = \fixwidehat A_k^{nT}$ having an equal or lower amount of (genuine) negativity, $N_m[\fixwidehat \rho] = N_m[\rho]$ for any bipartition $m$ and $N_G[\fixwidehat \rho] \leq N_G[\rho]$. Furthermore one has: 

a)~If $\rho$ is PPT across a bipartition $m$ then $\fixwidehat \rho$ can even be assumed to be PPT invariant across $m$, i.e., $\fixwidehat \rho = \fixwidehat \rho^{T_m} = \fixwidehat \rho^T$. Likewise, if $\rho$ is a PPT mixture then $\fixwidehat \rho$ can even be assumed to have a mixture of real PPT invariant states, $ \fixwidehat \rho=\sum p_m \fixwidehat \rho_m $ with $\fixwidehat \rho_m = \fixwidehat \rho_m^{T_m} = \fixwidehat \rho_m^T$ for all $m$. 

b)~If the observed distribution is invariant under arbitrary exchange of the parties, $P = V(\pi) P V(\pi)^\dag$ for all possible permutation $\pi$, then the underlying state and measurements can further be assumed to be permutationaly invariant, i.e., $\fixwidehat V(\pi)\, \fixwidehat \rho\, \fixwidehat V(\pi)^\dag = \fixwidehat \rho$ for all $\pi$ and $\fixwidehat A_i^{1} = \fixwidehat A_i^{2} = \ldots = \fixwidehat A_i^{n}$.
\end{proposition}

\begin{proof}
For the first statement we construct for any set of projectors $A^s$, $s=1,\ldots,n$ and any state $\rho$, another real state $\fixwidehat \rho=\fixwidehat \rho^T$ and real projectors $\fixwidehat A^s = \fixwidehat A^{sT}$ such that $\tr[\rho A^1\otimes A^2\otimes \ldots\otimes A^n]=\tr[\fixwidehat \rho \fixwidehat A^1 \otimes \fixwidehat A^2 \otimes \ldots\otimes \fixwidehat A^n]$ with $N_m[\fixwidehat \rho]=N_m[\rho]$ for any bipartition $m$ and $N_G[\fixwidehat \rho] \leq N_G[\rho]$.

Let us start with the projectors. For any chosen basis we can decompose the matrix $A = A_r + {i} A_i$ into a real, symmetric part $A_r = \mathbbm{Re}(A) = A_r^T$ and an imaginary, anti-Hermitian part $A_i = \mathbbm{Im}(A) = - A_i^T$. To this matrix we now associate the real and symmetric matrix 
\begin{equation}
\label{eq:barA}
\fixwidehat A = \mathbbm{1} \otimes A_r + Y \otimes A_i = \left[ \begin{array}{cc} A_r & A_i \\ -A_i & A_r \end{array} \right]=\fixwidehat A^T
\end{equation}
where $Y={i} \sigma_y$. This represents the well-known isomorphism between Hermitian and symmetric matrices, \eg, Ref.~\cite{cobook}. Note that the operator $\fixwidehat A$ acts on the enlarged Hilbert space $\mathcal{H}_{\fixwidehat A} = \mathcal{H}_{A^\prime} \otimes \mathcal{H}_A$ with $\mathcal{H}_{A^\prime} = \mathbbm{C}^2$. 
Since the original $A$ satisfies the projection identity $A^2 = A$ this provides the relations $A_r^2 - A_i^2 = A_r$ and $A_r A_i + A_i A_r = A_i$ and henceforth the projection identity for $\fixwidehat A$ by
\begin{align}
\nonumber
\fixwidehat A^2 &= (\mathbbm{1} \otimes A_r + Y \otimes A_i) (\mathbbm{1} \otimes A_r + Y \otimes A_i) \\
&= \mathbbm{1} \otimes (A_r^2 - A_i^2) + Y \otimes (A_rA_i + A_r A_i) = \fixwidehat A.
\end{align}
We employ this construction for all projectors $A^s \mapsto \fixwidehat A^s$. 

Next let us define the appropriate extension of the density matrix $\rho$ to $\fixwidehat \rho=\fixwidehat \rho_{A_1^\prime A_1 A_2^\prime A_2 \ldots A_n^\prime A_n}$ where the subscripts label the party to which the Hilbert space belongs. However, for ease of reading, we shall use a different ordering of the Hilbert spaces (namely, the primed ones followed by the unprimed ones) $\fixwidehat \rho_{A_1^\prime A_2^\prime\ldots A_n^\prime A_1A_2\ldots A_n} = \mathcal{R}(\rho)$, which could be undone $\fixwidehat \rho = V \fixwidehat \rho_{A_1^\prime \ldots A_n^\prime A_1 \ldots A_n} V$ by an appropriate permutation $V$ on Hilbert spaces. More precisely, this operator is given by
\begin{align}
\label{eq:barrho1}
	\fixwidehat \rho_{A_1^\prime A_2^\prime \ldots A_n^\prime A_1A_2\ldots A_n} &= 2^{-n}\,\mathbbm{Re}\left[(\mathbbm{1}-iY)^{\otimes \nn}\otimes\rho\right].
\end{align}

Let us first show that $\fixwidehat \rho$\, is positive semi-definite and has trace one as required for any density operator. To this end, first note that the right-hand-side of Eq.~\eqref{eq:barrho1} can be rewritten as a sum of two orthogonal parts, namely, $2^{-n-1}(\mathbbm{1}-iY)^{\otimes \nn}\otimes\rho$ and its transpose, where the orthogonality follows from the fact that $(\mathbbm{1}-iY)(\mathbbm{1}+iY)=0$. Since the set of eigenvalues remains unchanged under transposition, this means that the eigenvalues of $\fixwidehat \rho$ are precisely one-half of the respective eigenvalues of $2^{-n}(\mathbbm{1}-iY)^{\otimes \nn}\otimes\rho$, but with twice the multiplicity. However since $(\mathbbm{1}-iY)/2$ is just a rank-one projector one has that the non-vanishing eigenvalues of $2^{-n}(\mathbbm{1}-iY)^{\otimes \nn}\otimes K$ for any Hermitian matrix $K$ are precisely those of $K$. Thus, we obtain $\fixwidehat \rho \geq 0$, $\| \fixwidehat \rho \|_{\tr} =\| \rho \|_{\tr} =  1$ and $\| \mathcal{R}(\rho^{T_m})\|_{\tr} = \| \rho^{T_m} \|_{\tr} $ for any bipartition $m$.  
 
Next we need to show that these choices indeed preserve the expectation values. This is best seen using the state given by Eq.~(\ref{eq:barrho1}) and the measurement form of Eq.~(\ref{eq:barA}) together with the following identity 
\begin{align}
\nonumber
\tr_{A'}\left[(\mathbbm{1}-iY)\otimes \rho \right.&\left. (\mathbbm{1}\otimes A_r + Y\otimes A_i)\right]\\
\nonumber
=& \tr_{A'}\left[\mathbbm{1}\right] \rho\,A_r + \tr_{A'}\left[-iYY\right] \rho\,A_i  \\
\nonumber
&+ \tr_{A'}\left[Y\right] \rho\,A_i  + \tr_{A'}\left[-iY\right] \rho\, A_r \\
=& 2\rho\,A,
\end{align}
which follows from $Y^2=-\mathbbm{1}$ and $\tr(Y)=0$. Applying such identities when tracing out each auxiliary space yields
\begin{align}
\nonumber
\tr[\fixwidehat \rho  \fixwidehat A^1 \otimes \ldots \otimes \fixwidehat A^n]=&2^{-n}\, \mathbbm{Re}\left\{ \tr\left[V(\mathbbm{1}-iY)^{\otimes n}\otimes \rho V \times\right.\right.\\
\nonumber
&\ \ \ \left.\left.(\mathbbm{1} \otimes A^1_r + Y \otimes A^1_i)\otimes\ldots\right]\right\}\\
\nonumber
=& \mathbbm{Re}\left(\tr[ \rho\  A^1\otimes \ldots \otimes A^n]\right)\\
=& \tr[ \rho\  A^1\otimes \ldots \otimes A^n].
\end{align}

Next, we prove that the negativity remains constant $N_m[\rho] = N_m[\fixwidehat \rho]$ for any bipartition $m$. Thus we need the partial transposition of $\fixwidehat \rho$ with respect to $m$. For simplicity, we now consider partial transposition with respect to $\fixwidehat A$; the general treatment is completely analogous.  To this end, we remind that the partial transposition is a linear operation and thus from Eq.~\eqref{eq:barrho1} and the fact that $Y^T=-Y$, we have
\begin{align}
	\fixwidehat \rho_{A_1^\prime \ldots A_n^\prime A_1\ldots A_n}^{T_{A_1^\prime A_1}}
	&= 2^{-n}\,\mathbbm{Re}\left[(\mathbbm{1}+iY)\otimes(\mathbbm{1}-iY)^{\otimes n-1}\otimes\rho\right]\!,\nonumber\\
	&= U\, \mathcal{R}[ \rho^{T_{A_1}}]\,U^\dag
\end{align}
where $U=\sigma_x\otimes \mathbbm{1}^{n-1}\otimes\mathbbm{1}_{A_1\ldots A_n}$ is a unitary matrix and we have made use of the fact that $\sigma_x\,Y\,\sigma_x=-Y$.
Since a unitary does not change the eigenvalues and using the previously mentioned invariance of the trace-norm we get $\|\fixwidehat \rho^{T_{\fixwidehat A}}\|_{\tr} = \| \mathcal{R}(\rho^{T_A})\|_{\tr} = \| \rho^{T_A}\|_{\tr}$, and hence the statement that the negativity is unchanged. This statement holds for an arbitrary bipartition $m$.
Finally, suppose that $\rho = \sum p_m \rho_m$ is the optimal decomposition for the original state $\rho$ in the definition of the genuine negativity, \ie, $N_G[\rho] = \sum p_m N_m[\rho_m]$. Note that because the map $\rho \mapsto \fixwidehat \rho$ is linear and preserves positivity one knows that $\fixwidehat \rho = \sum_m p_m \fixwidehat \rho_m$ is a valid decomposition of $\fixwidehat \rho$ into states attributed to different bipartitions $m$. Since the genuine negativity is defined by the minimum over all such decompositions, and because $N_m[\fixwidehat \rho_m]=N_m[\rho_m]$ one readily obtains $N_G[\fixwidehat \rho] \leq  \sum p_m N_m[\fixwidehat \rho_m] = N[\rho]$, which proves the first part of the proposition.

Henceforth, we shall work with the real state $\fixwidehat \rho$ and the corresponding real measurement operators $\fixwidehat A_i^s=\fixwidehat A_i^{sT}$, $s=1,\ldots,n$ that can produce the same correlations. Assuming this, statement a) is obtained as follows: If the original state $\rho$ is PPT across a bipartition $m$, then the real state $\fixwidehat \rho$ has to be PPT as well. Therefore $\fixwidehat \rho^{T_m} \geq 0$ can be considered as another quantum state, which, similarly would produce the correct observations since all measurements are real, 
\begin{align}
\nonumber
P &= \tr[\fixwidehat \rho \fixwidehat A_i^1 \otimes \fixwidehat A_j^2 \otimes \ldots \otimes \fixwidehat A_k^n]\\
\nonumber
&= \tr[\fixwidehat \rho (\fixwidehat A_i^1 \otimes \fixwidehat A_j^2 \otimes \ldots \otimes \fixwidehat A_k^n)^{T_m}] \\
&= \tr[\fixwidehat \rho^{T_m} \fixwidehat A_i^1 \otimes \fixwidehat A_j^2 \otimes \ldots \otimes \fixwidehat A_k^n].
\end{align}
Here, between the first and the second line we transposed the measurement operator of all parties belonging to the bipartition $m$. Thus also the equal-weight mixture $\fixwidehat \rho^{\rm ave} = (\fixwidehat \rho + \fixwidehat \rho^{T_m})/2$ is a real state that is PPT invariant with respect to $m$ and that yields the same correlations. A similar statement holds for PPT mixtures with partial transposition applied to each state $\rho_m$ in the decomposition.

Statement b) follows from a similar argument. For simplicity, we provide the proof below only for $n=2$ parties; the generalization to an arbitrary number of parties is analogous. Let us suppose that $\fixwidehat \rho=\fixwidehat \rho_{AB}$ and that this state as well as $\fixwidehat A_i,\fixwidehat B_j$ are not yet symmetric under exchange of the parties.\footnote{At this point we can already assume that the local Hilbert spaces are isomorphic $\mathcal{H}_{A} \cong \mathcal{H}_{B}$ by an appropriate embedding of a possibly smaller space in higher dimension. In this way the swap operator $V=V(\pi_{AB})$ is well defined.} Let us add appropriate local auxiliary states $\ket{0},\ket{1}$ (on system $A'$, $B'$ respectively) to signal whether the state is swapped or not, and consider the convex combination $\fixwidehat \rho^{\rm sym}_{\fixwidehat A \fixwidehat B}$ given by
\begin{align}
\nonumber
\fixwidehat \rho^{\rm sym}_{\fixwidehat A \fixwidehat B} \! &= \frac{1}{2} \left[ \fixwidehat \rho_{AB} \! \otimes \! \ket{01}\!_{A^\prime \! B^\prime}\!\bra{01}  + V \fixwidehat \rho_{AB} V^\dag \!\otimes\! \ket{10}\!_{A^\prime \! B^\prime}\!\bra{10} \right] \\
\nonumber
&=\frac{1}{2} \left[ \fixwidehat \rho_{AB} \!  \otimes \!  \ket{01}\!_{A^\prime \! B^\prime}\!\bra{01} \right. \\ & \label{eq:sym} \:\:\:\:\:\:\:\:\:\:\:\: \left.+  {\fixwidehat{V} ( \fixwidehat \rho_{AB}\!  \otimes\!  \ket{01}\!_{A^\prime \! B^\prime}\!\bra{01})  \fixwidehat{V}^{\dag}}\right]
\end{align}
where $\fixwidehat{V}=V(\pi_{\fixwidehat A \fixwidehat B})$ now denotes the swap operator on $\fixwidehat A = AA^\prime$ and $\fixwidehat B = BB^\prime$. Similarly for the measurements 
\begin{equation}
\fixwidehat A^{\rm sym}_i = \fixwidehat B^{\rm sym}_i = \fixwidehat A_i \otimes \ket{0}\bra{0} + \fixwidehat B_i \otimes \ket{1}\bra{1},
\end{equation}
such that one finally can check that this is a valid quantum representation with
\begin{equation}
P = [P+VPV^\dag]/2 = \tr[\fixwidehat \rho_{AB}^{\rm sym} \fixwidehat A_i^{\rm sym} \otimes \fixwidehat A_j^{\rm sym}].
\end{equation}
From the structure given by Eq.~(\ref{eq:sym}) one observes that $\fixwidehat \rho_{\fixwidehat A \fixwidehat B}^{\rm sym}$ is indeed invariant under the swap operator $\fixwidehat{V}$. Finally since the negativity is convex and symmetric~\cite{horodecki10a} one obtains $N[\fixwidehat \rho_{\fixwidehat A \fixwidehat B}^{\rm sym}] \leq N[\fixwidehat \rho_{AB}]$, while the LOCC monotonicity by measuring the primed systems first provides $N[\fixwidehat \rho_{\fixwidehat A \fixwidehat B}^{\rm sym}] \geq N[\fixwidehat \rho_{AB}]$ such that the negativity is indeed invariant. 
\end{proof}

Let us briefly remark how Proposition~\ref{Proposition:RealStuff} reduces the number of free parameters for the semidefinite programs.

Firstly, note that since the underlying state and both measurements can be chosen to be real, the matrix of moments can be assumed to be symmetric $\chi[\rho] = \chi[\rho]^T$. For the multipartite PPT mixture question of Eqs.~\eqref{eq:ppt_genuine_tsirelson},\eqref{eq:dvi_genuine_negativity}, this property can be applied to every bipartition, \ie, $\chi[\rho_m] = \chi[\rho_m]^T$. Concerning the negativity finally, this property can be further imposed for the operators appearing in the variational formulation, since a solution given by $\rho_m = \sigma_m^+ - \sigma_m^-$ with $(\sigma_m^{\pm})^{T_m} \geq 0$ would imply the alternative solution $\rho_m =\rho_m^T =  (\sigma_m^+)^T - (\sigma_m^-)^T$ with $[(\sigma_m^\pm)^T]^{T_m} = [(\sigma_m^\pm)^{T_m}]^T \geq 0$, and hence also its equal-weight mixture. Therefore in the negativity estimation as given by the semidefinite programs of Eqs.~($4$),(\ref{eq:dvi_genuine_negativity}), we can set additionally that $\chi[\rho_m] = \chi[\rho_m]^T$ and $\chi[\sigma_m^{\pm}]= \chi[\sigma_m^{\pm}]^T$.

The statement a) also simplifies the computation of the respective PPT Tsirelson bounds as given by Eq.~($3$) or Eq.~(\ref{eq:ppt_genuine_tsirelson}). Since the state can be assumed to be real and PPT invariant, this provides the symmetry $\chi[\rho_m] = \chi[\rho_m]^T = \chi[\rho_m]^{T_{\bar m}} \geq 0$ for all $m$.

Finally, whenever one considers a symmetric Bell inequality~\cite{bancal10}, \ie, satisfying $I_{abxy} = I_{bayx}$ in the bipartite case, one can impose this symmetry also for the corresponding distribution $P(a,b|x,y) = P(b,a|y,x)$, such that b) of Proposition~\ref{Proposition:RealStuff} gets relevant. If both local maps in the construction of $\chi[\rho]$ have an equal number of moments, then $\chi[\rho]$ can be assumed to be invariant $V_{\bar A \bar B} \chi[\rho] V_{\bar A \bar B}^\dag = \chi[\rho]$ under the swap operator $V_{\bar A \bar B}$. This symmetry in particular helps in the higher levels of the hierarchy.

\section{Convergence of respective PPT characterizations}

As emphasized in the main text, the method to approximate PPT Tsirelson bounds via Eq.~($3$) converges to a description of the commutative set in the limit of an infinite number of moments. A similar statement holds for the multipartite case and PPT mixtures. To show this, we start this section by formulating the commutative version of the respective quantum representation. We then proceed to show the equivalence between this commutative and the tensor product version in the case of a finite-dimensional quantum representation and, at last, prove the convergence. We stress that this convergence follows from the convergence of the NPA hierarchy~\cite{navascues08a}. 

At first let us motivate this distinction between the tensor product and commutative version: In quantum information we are commonly used to employ tensor products for measurements on different subsystems, \ie, $\tr[ \rho_{AB} A \otimes B]$. In contrast, in algebraic field theory for instance, measurements on different parts are described by commuting operators $[\tilde A, \tilde B]=0$, both acting already on the bipartite state $\tilde \rho$ (therefore we use the tilde) such that expectation values become $\tr[\tilde \rho \tilde A \tilde B]$. If one identifies $\tilde A = A \otimes \mathbbm{1}, \tilde B = \mathbbm{1} \otimes B$ and $\tilde \rho = \rho_{AB}$ then one sees that all expectation values using tensor products can be recovered with commuting observables. The converse question, however, is still open and is known as Tsirelson's problem~\cite{scholz08a,navascues12a,junge11,fritz12a}. 

For clarity, we start with the translation of the bipartite PPT constraint $\rho^{T_A} \geq 0$ in commutative terms. We employ the simplified notions $A_i,B_j,\ldots$ introduced in the main text for measurement operators, and $P=\tr[\dots A_i \dots B_j \ldots]$ to indicate the functional form of the observed probability distribution.

\begin{definition}
i) A bipartite probability distribution $P$ has a \emph{tensor product PPT quantum representation} if there exists a real state $\rho_{AB}$, local real measurements operators $A_i, B_j$ such that $P=\tr[\rho_{AB} A_i \otimes B_j]$ and that  $\rho_{AB} = \rho_{AB}^{T_A}$ is PPT invariant. 

ii) A bipartite distribution $P$ has a \emph{commutative PPT quantum representation} if there exists a real state $\tilde \rho$, commuting real measurement operators $\tilde A_i, \tilde B_j$, $[\tilde A_i, \tilde B_j] = 0$, such that $P=\tr[\tilde \rho \tilde A_i \tilde B_j]$ and for all possible products, indexed by $i=(i_1,i_2,\dots,i_n)$, $\tilde{A}(i) = \tilde A_{i_1} \tilde A_{i_2} \dots \tilde A_{i_n}$, and similarly for $\tilde B$, it holds 
\begin{equation}
\label{eq:ppt_comm}
\tr[\tilde \rho \tilde A(i) \tilde B(j)] = \tr[ \tilde \rho \tilde A(i^T) \tilde B(j)]
\end{equation}
with $i^T=(i_n,i_{n-1},\dots,i_1)$.
\end{definition}

\begin{proposition}\label{Proposition:TensorProduct}
Any bipartite probability distribution which has a tensor product PPT quantum representation also has a commutative PPT quantum representation. The converse holds at least if the commutative PPT quantum representation is finite-dimensional.
\end{proposition}

\begin{proof} 
Let us start with the direction i) to ii). The formulation i) is already simplified using Proposition~\ref{Proposition:RealStuff}. If we now interpret $\tilde A = A \otimes \mathbbm{1}, \tilde B = \mathbbm{1} \otimes B$ and $\tilde \rho = \rho_{AB}$ then these operators naturally fulfil the requirements of being real, that their commutator vanishes and that they provide the correct expectation values. If one employs the PPT invariance $\rho_{AB} = \rho_{AB}^{T_A}$ and $A_i = A_i^T$ one obtains the condition given by Eq.~(\ref{eq:ppt_comm}) via
\begin{align}
\nonumber
\tr[\tilde \rho \tilde{A}(i) \tilde{B}(j)] &= \tr[ \rho_{AB} A(i)\! \otimes\! B(j)] = \tr[ \rho_{AB}^{T_A} A(i) \! \otimes \! B(j)] \\ 
\nonumber
&=  \tr[ \rho_{AB} (A_{i_1} A_{i_2} \dots A_{i_n})^T \!\otimes \!B(j)] \\
\nonumber
&= \tr[ \rho_{AB} (A_{i_n} A_{i_{n-1}}\! \dots A_{i_1})\! \otimes B(j)] \\ 
\nonumber
&= \tr[ \rho_{AB} A(i^T)\! \otimes\! B(j)] \\
&= \tr[\tilde \rho \tilde{A}(i^T) \tilde{B}(j)].
\end{align}
This finishes the proof in the direction i) to ii).

Let us now turn to the converse. Thus suppose that $\tilde \rho$ and $\tilde A_i, \tilde B_j$, all acting on the finite-dimensional Hilbert space $\mathcal{\tilde H}$, are given as in ii). From these operators we now construct local measurement operators and a bipartite quantum state which is PPT; the statement that they even have the extra symmetries follows from point a) of Proposition~\ref{Proposition:RealStuff}. This construction is very analogous to the corresponding one of the standard Tsirelson problem. However, in order to show that the conditions given by Eq.~(\ref{eq:ppt_comm}) finally prove that the constructed bipartite state is PPT one has to keep track which operators can be built up by products of $\tilde A_i, \tilde B_j$ and linear combinations of them, or, more precisely, by the operators set 
\begin{equation}
\mathcal{Q} = \big\{ Q \in \mathbbm{B}(\mathcal{\tilde H}) \big| Q = \sum_{ij} c_{ij} \tilde A(i) \tilde B(j), c_{ij} \in \mathbbm{C} \big \}.
\end{equation}
From the given finite-dimensional commuting Hermitian operators $\tilde A_i, \tilde B_j$ one can infer the following decompositions
\begin{align}
\label{eq:decomp_H1}
\mathcal{\tilde H} &= \oplus_k \mathcal{H}_{AB,k}  \otimes \mathcal{K}_{k} \\ 
\label{eq:decomp_H2}
&= \oplus_k ( \oplus_l \mathcal{H}_{A,kl} \otimes \mathcal{H}_{B,kl} ) \otimes \mathcal{K}_{k}, 
\end{align}
and representations
\begin{align}
\label{eq:decomp_algebra}
\mathcal{Q} &= \oplus_k \mathbbm{B}( \mathcal{H}_{AB,k} ) \otimes \mathbbm{1},\\
\label{eq:A}
\tilde A_i &= \oplus_k A_{i,k} \otimes \mathbbm{1} =\oplus_k ( \oplus_l \tilde A_{i,kl} \otimes \mathbbm{1} ) \otimes \mathbbm{1},  \\
\label{eq:B}
\tilde B_j &= \oplus_k B_{j,k} \otimes \mathbbm{1} =\oplus_k ( \oplus_l \mathbbm{1} \otimes \tilde B_{j,kl}) \otimes \mathbbm{1}. 
\end{align}
These statements can be inferred for instance from Theorem A.7 of Ref.~\cite{doherty08a}: The original structure given by Eq.~(\ref{eq:decomp_H1}) comes from the decomposition of (the finite-dimensional $C^*$-algebra) $\mathcal{Q}$. The tensor product structure of $\mathcal{H}_{AB,k}$, included in Eq.~(\ref{eq:decomp_H2}), originates from the same argument as in the standard Tsirelson problem, applied to the operators $\tilde A_{i,k}, \tilde B_{j,k}$ acting on $\mathcal{H}_{AB,k}$ which still commute. (In addition to Theorem A.7 of Ref.~\cite{doherty08a} one needs to know that $\tilde B_{j,k}$ are elements of the commutant of $C^*$-algebra generated by $\{ \tilde A_{i,k} \}$). 

Since it will become important shortly, let us stress the meaning of Eq.~(\ref{eq:decomp_algebra}): It states that any operator of this form is in the set $\mathcal{Q}$. Note that this structure includes, in particular, all operators $\oplus_k [\oplus_l Q_{kl}] \otimes \mathbbm{1}$ with $Q_{kl}$ arbitrary. Thus, also the very special operators where all of these operators $Q_{kl}$ are zero except for one particular pair, maybe $k=k^\prime,l=l^\prime$. Hence by appropriately chosen coefficients $c_{ij}[Q_{k^\prime l^\prime}]$ it is possible to build up 
\begin{equation}
\label{eq:Q_kl}
\sum_{ij} c_{ij}[Q_{k^\prime l^\prime}] \tilde A(i) \tilde B(j)= Q_{k^\prime l^\prime} \otimes \mathbbm{1}_{\mathcal{K}_{k^\prime}}
\end{equation}
with $Q_{k^\prime l^\prime}$ being an arbitrary operator acting on $\mathcal{H}_{AB,kl}$.

Next let us construct the bipartite state. In the following we denote by $\Pi_{kl}$ the projections onto $\mathcal{H}_{AB,kl}\otimes \mathcal{K}_k$. From the structure of the Hilbert space and the measurements one obtains the identity
\begin{equation}
\label{eq:mixture}
\tr[\tilde \rho \tilde A_i \tilde B_j]\! =\! \sum_{kl} \tr_{AB,kl}[\tr_{\mathcal{K}_k}(\Pi_{kl} \tilde \rho \Pi_{kl}) \tilde A_{i,kl} \otimes \tilde B_{j,kl}].
\end{equation}
If one defines $p_{kl} \tilde \rho_{AB,kl}= \tr_{\mathcal{K}_k}(\Pi_{kl}\tilde \rho \Pi_{kl})$ one can interpret Eq.~(\ref{eq:mixture}) as a classical mixture of bipartite states on different subsystems. Via appropriate local auxiliary system one can combine this into a single bipartite system 
\begin{align}
\label{eq:rhoAB_tp}
\rho_{AB} &= \sum p_{kl} \tilde \rho_{AB,kl} \otimes \ket{kl,kl}_{A^\prime B^\prime}\bra{kl,kl}, \\
A_i &= \sum \tilde A_{i,kl} \otimes \ket{kl}_{A^\prime}\bra{kl},\\
B_j &= \sum \tilde B_{j,kl} \otimes \ket{kl}_{B^\prime}\bra{kl}, 
\end{align}
such that $\tr[\tilde \rho \tilde A_i \tilde B_j] = \tr[\rho_{AB} A_i \otimes B_j]$ following Eq.~(\ref{eq:mixture}). That $\rho_{AB}$ is indeed a valid quantum state, while $A_i, B_j$ are correct operators describing measurements follows from their structure. 

Thus we are left to show that $\rho_{AB}$ is PPT, which, by its form given by Eq.~(\ref{eq:rhoAB_tp}), is equivalent to $\tilde\rho_{AB,kl}$ being PPT for all $k,l$. A finite-dimensional operator $\tilde \rho_{AB,kl}^{T_A} \geq 0$ is positive semidefinite if and only if for all operators $Q_{kl} \in \mathbbm{B}(\mathcal{H}_{AB,kl})$ it holds that $\tr[\tilde \rho_{AB,kl}^{T_A} Q_{kl}Q_{kl}^\dag] \geq 0$~\cite{miranowicz09a}. In the remaining we show that this inequality holds due to $\tr[\tilde \rho QQ^\dag] \geq 0$ for all operators $Q\in\mathcal{Q}$ and the extra conditions given by Eq.~(\ref{eq:ppt_comm}). Here it is important to stress that $\mathcal{Q}$ indeed includes all operators $Q_{kl} \in \mathbbm{B}(\mathcal{H}_{AB,kl})$ as remarked previously. Using $Q \in \mathcal{Q}$ with $c_{ij}=c_{ij}[Q_{k^\prime l^\prime}]$ as given in Eq.~(\ref{eq:Q_kl}) for fixed $k^\prime,l^\prime$, $\tilde A_{kl}(i) = \tilde A_{i_1,kl}\dots \tilde A_{i_n,kl}$ and a similar shorthand for $\tilde B_{kl}(j)$, one obtains
\begin{align}
\nonumber
\tr[&\tilde \rho QQ^\dag] =\sum  c_{ij} c^*_{uv} \tr[\tilde \rho \tilde A(i)\tilde B(j)(\tilde A(u) \tilde B(v))^\dag] \\ 
\nonumber&=\sum c_{ij} c^*_{uv} \tr[\tilde \rho \tilde A(i)\tilde A(u^T)\tilde B(j)\tilde B(v)^\dag]  \\
\nonumber&= \!\!\!\!\!\!\!\!\!\!\!\!\! \stackrel{\textrm{Eq.}~(\ref{eq:ppt_comm})}{\phantom{=}} \sum c_{ij} c^*_{uv} \tr[\tilde \rho \tilde A(u)\tilde A(i^T)\tilde B(j)\tilde B(v)^\dag] \\
\nonumber&= \sum c_{ij} c^*_{uv} \tr[\tilde \rho (\tilde A(i)\tilde A(u)^\dag)^T \tilde B(j)\tilde B(v)^\dag ] \\
\nonumber&= \sum c_{ij} c^*_{uv} p_{kl}\\
\nonumber
& \:\: \times \tr_{AB,kl}[ \tilde \rho_{AB,kl}  (\tilde A_{kl}(i) \tilde A_{kl}(u)^\dag )^T \otimes \tilde B_{kl}(j) \tilde B_{kl}(v)^\dag  ] \\
\nonumber&= \sum c_{ij} c^*_{uv} p_{kl}\\
\nonumber
& \:\: \times \tr_{AB,kl}[ \tilde \rho_{AB,kl}^{T_A}  \tilde A_{kl}(i) \tilde A_{kl}(u)^\dag \otimes \tilde B_{kl}(j) \tilde B_{kl}(v)^\dag  ] \\\label{eq:pos_pt}
&= p_{k^\prime l^\prime} \tr_{AB,k^\prime l^\prime}[ \tilde \rho_{AB,k^\prime l^\prime}^{T_{A}} Q_{k^\prime l^\prime} Q_{k^\prime l^\prime}^\dag ] \geq 0.
\end{align}
Here we employ the following: Because $\tilde A_i$ commutes with every $\tilde B_j$ this property naturally extends to products $[\tilde A(i),\tilde B(j)]=0$. Moreover, since the measurement operators $\tilde A_i$ can be taken real and symmetric, it follows that $\tilde A(i)^\dag = \tilde A(i^T)$ and similarly $\tilde A(k) \tilde A(i^T)=[\tilde A(i)\tilde A(k)^\dag]^T$. As indicated the identity given by Eq.~(\ref{eq:ppt_comm}) is employed in the third line. Afterwards one uses once more the block decomposition as already shown previously, and the identity $\tr[ X Y^{T_A}] = \tr[ X^{T_A} Y]$. The final inequality stems from the positivity of the original density operator $\tilde \rho$. This finishes the proof of Proposition~\ref{Proposition:TensorProduct}.
\end{proof}

Let us now see how the property of a PPT mixture can be cast into a commutative version. The equivalence of both formulations in the finite case is shown similarly as in the bipartite case.

\begin{definition}
i) An $n$-partite probability distribution $P$ has a \emph{tensor product PPT mixture quantum representation} if there exists a real $n$-partite state $\rho$, local real measurement operators $A_i^1,A_j^2,\ldots,A^n_k$ such that $P=\tr[\rho A^1_i \otimes A_j^2 \otimes \ldots  \otimes A_k^n]$ and $\rho = \sum_m p_m\rho_m$ being a convex combination of real, PPT invariant states $\rho_m = \rho_m^{T_m}$ for all bipartitions $m$.

 ii) An $n$-partite probability distribution $P$ has a \emph{commutative PPT mixture quantum representation} if there exists a real state $\tilde \rho$, commuting real measurement operators $\tilde A^1_i, \tilde A^2_j,\ldots, \tilde A^n_k$ such that $P=\tr[\tilde \rho \tilde A^1_i \tilde A^2_j\ldots \tilde A_k^n]$ and $\rho = \sum_{m} p_m \rho_m$ is a convex combination of real ``PPT states'' $\rho_m$ for all possible formal bipartitions $m$, \ie, each state $\rho_m$ must satisfy that for all possible products of operators from partition $m$, indexed by $i^{(m)} = (i_1^{k_1},i_2^{k_2},\ldots,i_r^{k_r})$ with $k_s \in m$ for $s=1,\ldots,r$, $\tilde A(i^{(m)})=\tilde A_{i_1}^{k_1}\tilde A_{i_2}^{k_2}\ldots\tilde A_{i_r}^{k_r}$ and similar for products for its complement $m_c$ it holds 
\begin{equation}
\label{eq:ppt_comm_multi}
\tr[\tilde \rho_m \tilde A(i^{(m)}) \tilde A(j^{(m_c)})] = \tr[ \tilde \rho_m \tilde A( i^{(m)T}) \tilde A(j^{(m_c)})]
\end{equation}
with $i^{(m)T}=(i_r^{k_r},i_{r-1}^{k_{r-1}},\dots,i^{k_1}_1)$.
\end{definition}

\begin{proposition}
Any probability distribution which has a tensor product PPT mixture quantum representation also has a commutative PPT mixture quantum representation. The converse holds at least if the commutative PPT mixture quantum representation is finite-dimensional.
\end{proposition}

\begin{proof} 
The direction from $i)$ to $ii)$ holds by similar arguments as given in the bipartite case. 

For the converse direction one proceeds analogously as in the proof of Proposition~$2$. From the given measurement operators one can obtain, similar to Eqs.(\ref{eq:decomp_H1})-(\ref{eq:B}), a decomposition of the underlying Hilbert space given by
\begin{equation}
\label{eq:decompHmulti}
\mathcal{H} = \oplus_k (\oplus_l \mathcal{H}_{A_1,kl} \otimes \mathcal{H}_{A_2,kl} \otimes \ldots \otimes \mathcal{H}_{A_n,kl} ) \otimes \mathcal{K}_{k},
\end{equation}
its measurement operators and respective operator set $\mathcal{Q}$, \ie, the set of all operators generated by linear combinations of products of the measurements~\cite{doherty08a}. This shows that all measurements $\tilde A_i^s$ from system $s$ act only non-trivial on the part $\mathcal{H}_{A_s,kl}$, in analogy to the bipartite case as given by Eqs.~(\ref{eq:A}),(\ref{eq:B}). Via the projector $\Pi_{kl}$ onto the subspace $\mathcal{H}_{A_1\ldots A_n,kl}\otimes \mathcal{K}_k$ one arrives again at the tensor product form  
\begin{align}
\nonumber
\tr &[\tilde \rho \tilde A^1_i \ldots \tilde A_j^n] \\ &= \sum_{kl} p_{kl} \tr_{A_1\dots A_n,kl} [ \tilde \sigma_{kl} \tilde A_{i,kl}^1 \otimes \ldots \otimes A_{j,kl}^n ].
\end{align}
with $p_{kl} \tilde \sigma_{kl} = p_{kl} \tilde \sigma_{A_1\dots A_n,kl} = \tr_{\mathcal{K}_{k}} [ \Pi_{kl} \tilde \rho \Pi_{kl}]$. Via auxiliary states this classical mixture can be turned into a single multipartite state, \ie, $\rho = \sum_{kl} p_{kl} \tilde \sigma_{kl}v \otimes \ket{kl}\bra{kl}^{\otimes n}$ and measurements $A_i^s = \sum_v \tilde A_{i,kl}^s \otimes \ket{kl}\bra{kl}$ for each system $s$. 

At last one needs to show that $\rho$ is a PPT mixture, which amounts to verifying that each $\tilde \sigma_{kl}$ is a PPT mixture. This decomposition is given by the original expansion $\tilde \rho = \sum_m \tilde \rho_m$ into real (in this case un-normalized) states $\tilde \rho_m$ which satisfy the property given by Eq.~(\ref{eq:ppt_comm_multi}). If we set $\rho_{m,kl} = \tr_{\mathcal{K}_{k}} [ \Pi_{kl} \tilde \rho_m \Pi_{kl}] $ then one obtains a decomposition
\begin{equation}
p_{kl} \tilde \sigma_{kl} = \sum_m \tilde \rho_{m,kl},
\end{equation} 
which verifies that $\tilde \sigma_{kl}$ is a PPT mixture if one can show that $\tilde \rho_{m,kl}$ is PPT with respect to bipartition $m$. Since we can interpret $\tilde \rho_{m,kl}$ as a bipartite system on $m$ and its complement $\bar m$ this boils down to the bipartite argument: The condition $\tilde \rho_{m,kl}^{T_m} \geq 0$ is equivalent to $\tr[ \tilde \rho_{m,kl}^{T_m} Q_{kl}Q_{kl}^\dag] \geq 0$ for all operators $Q_{kl} \in \mathbbm{B}(\mathcal{H}_{A_1\ldots A_n,kl})$. Since any operator $Q_{kl}$ can again be written as a linear combination of products of the measurement operators, the identities given by Eq.~(\ref{eq:ppt_comm_multi}) verify $\tr[ \tilde \rho_{m,kl}^{T_m} Q_{kl}Q_{kl}^\dag] = \tr[ \tilde \rho_{m,kl} Q_{kl}Q_{kl}^\dag]\geq 0$, where the inequality results from the positivity of $\tilde \rho_m $. This concludes the proof.~\hfill 
\end{proof}

Finally let us state the convergence result. As already pointed out this is a corollary of the original convergence proof for the NPA hierarchy~\cite{navascues08a}. In this part we use the notation $\chi^{(\ell)}$ to refer to the matrix of moments of level $\ell$, \ie, having for each party a product of up to $\ell$ local operators in the construction. We first prove the convergence of the PPT hierarchy for the bipartite case. This proof then directly generalizes to a PPT hierarchy for a single bipartition $m$ in the multipartite case, which finally helps to prove the convergence of the PPT mixture in the end. Note that the condition $\chi^{(\ell)}[P_,u]_{\tr}=1$ does not appear since we explicitly talk about probability distributions.

\begin{proposition}\label{Proposition:PPT}
If a bipartite distribution $P$ satisfies $\chi^{(\ell)}[P,u] = \chi^{(\ell)}[P,u]^{T_{\bar A}} = \chi^{(\ell)}[P,u]^{T} \geq 0$ via suitable $u$ for all levels $\ell$, then it has a (potentially infinite-dimensional) commutative PPT quantum representation.
\end{proposition}

\begin{proof} 
In the following we use the label $s,t,v,w$ to denote sequences of measurement operators, \eg, $s=(s_1,s_2, \dots, s_n)$ with $A(s) = A_{s_1} A_{s_2} \dots A_{s_n}$. However, if we employ $i,j$ it refers to a single index, and $i=0$ or $j=0$ should correspond to the identity operator. The structure of $\chi^{(\ell)}$, as given by Eq.~($1$) for $\ell=1$, is such that the matrix elements correspond to expectation values
\begin{equation}
\braket{st|\chi^{(\ell)}|vw} = \tr\{\rho [A(s)\otimes B(t)][A(v)\otimes B(w)]^\dag\}.
\end{equation}
The identities given by the described partial information can be parsed as $\braket{st|\chi^{(\ell)}|vw} = \braket{s' t'|\chi^{(\ell)}|v' w'}$ if $A(s)A(v)^\dag \otimes B(t)B(w)^\dag = A(s')A(v')^\dag \otimes B(t')B(w')^\dag$ by the properties satisfied by operators, which are listed as $1)-3)$ in the main text. The two extra requirements provide the additional relations $\braket{st|\chi^{(\ell)}|vw} = \braket{vt|\chi^{(\ell)}|sw} = \braket{vw|\chi^{(\ell)}|st}$.

Next let us summarize the convergence idea of the NPA hierarchy~\cite{navascues08a}. Suppose that $\chi^{(\ell)}[P,u] = \chi^{(\ell)}[P,u]^{T_{\bar A}} = \chi^{(\ell)}[P,u]^{T}\geq 0$ exists for each level $\ell$, then there is also a limit $\chi^{(\infty)} \geq 0$ which satisfies all the listed identities. For this infinite-dimensional matrix one can associate a set of vectors $\{ \ket{e_{st}} \}$ with $\braket{st|\chi^{(\infty)}|vw}= \braket{e_{st}|e_{vw}}$. Via this set one now defines the measurements and the corresponding state. The measurements are $\hat A_i = \Proj (  \{ \ket{e_{st}} : s_1 = i \})$ where $\Proj$ stands for the projector onto the corresponding set. Similarly one defines $\hat B_j = \Proj (  \{ \ket{e_{st}} : t_1 = j \})$, while the state is given by $\hat \rho = \ket{e_{00}}\bra{e_{00}}$. As shown in detail in Ref.~\cite{navascues08a} these measurements satisfy $\hat A_i \ket{ e_{st} } = \ket{ e_{\tilde st}}$ with $\tilde s=(i,s)$ and similar for $\hat B_j$ due to the linear constraints. These choices reproduce the observed expectation values since
\begin{align}
\nonumber
P_{ij} &= \braket{e_{00}|e_{ij}} = \braket{e_{00}| \hat A_i \hat B_j | e_{00}} \\ & = \tr[ \ket{e_{00}}\bra{e_{00}} \hat A_i \hat B_j] = \tr[ \hat \rho \hat A_i \hat B_j].
\end{align}
Moreover, besides being true projectors, the operators $\hat A_i, \hat B_j$ fulfil the projector identity and commute. 

Now let us add the remaining properties. First, let us point out that since $\chi^{(\infty)}=\chi^{(\infty)T}$ the set $\{ \ket{e_{st}} \}$ can be chosen to be a set of real vectors. Therefore, the state and also the measurements are real symmetric, which implies in particular $\hat A(s)^\dag = \hat A(s^T)$. Thus we are left to show the property given by Eq.~(\ref{eq:ppt_comm}). However, this follows from the extra requirement  $\braket{st|\chi^{(\infty)}|vw} = \braket{vt|\chi^{(\infty)}|sw}$ since 
\begin{align}
\nonumber
\tr&[ \hat \rho \hat A(s) \hat B(t) ] =  \braket{ e_{00} | \hat A(s) \hat B(t) | e_{00}}  \\ 
\nonumber 
&= \braket{e_{00}| e_{st}} = \braket{e_{s0}| e_{0t}} = \braket{ e_{00} | \hat A(s)^\dag \hat B(t) | e_{00}}  \\ 
\label{eq:help1}
&= \braket{ e_{00} | \hat A(s^T) \hat B(t) | e_{00}}  = \tr[ \hat \rho \hat A(s^T) \hat B(t) ].
\end{align}
This concludes the convergence proof. \hfill 
\end{proof}

Let us remark that the proof works equivalently in the multipartite case if one is only interested in a single bipartition $m$: If an $n$-partite distribution $P$ satisfies $\chi^{(\ell)}[P,u] = \chi^{(\ell)}[P,u]^{T_{\bar m}} = \chi^{(\ell)}[P,u]^{T} \geq 0$ via suitable $u$ for all levels $\ell$, then it has a (potentially infinite-dimensional) commutative PPT quantum representation with respect to partition $m$. As for the bipartite case one uses the convergence statement in the multipartite case to show existence of a respective state $\hat \rho$ and commuting measurements $\hat A_i^1,\ldots,\hat A_k^n$. The additional linear requirements given by Eq.~(\ref{eq:ppt_comm_multi}) for $m$ are shown similarly as in Eq.~(\ref{eq:help1}) by using the identities $\chi^{(\ell)}[P,u] = \chi^{(\ell)}[P,u]^{T_{\bar m}}=\chi^{(\ell)}[P,u]^{T}$.

Before we continue we like to clarify another point: As seen in the proof of Proposition~\ref{Proposition:PPT} one always obtains a pure state $\tilde \rho = \ket{\psi}\bra{\psi}$ in the commutative version. However, since any pure PPT state of a tensor product Hilbert space is separable and thus possesses a LHV model, one could ask whether this implies that any (finite-dimensional) PPT state must necessarily satisfy all Bell inequalities. This would prove the Peres conjecture, at least, in the finite-dimensional case. However this is not true. Indeed, even if the state in the commutative version is pure, the related bipartite, tensor product state following the procedure outlined in the proof of Proposition~\ref{Proposition:TensorProduct} does not have to be. More precisely, this bipartite state $\rho_{AB}$ is given by Eq.~(\ref{eq:rhoAB_tp}), where the partial trace over $\mathcal{K}_k$ in $\rho_{AB,kl}=\tr_{\mathcal{K}_k}(\Pi_{kl}\ket{\psi}\bra{\psi} \Pi_{kl})$ results in a mixed PPT state. This extra system $\mathcal{K}_k$ can be considered as a purifying system of each bipartite state $\rho_{AB,kl}$ to $\Pi_{kl}\ket{\psi}$, which would then be pure if one translates it into commutative terms. Finally note that the purifying systems $\mathcal{K}_k$ would vanish if the commuting measurements $\tilde A_i$, $\tilde B_j$ would generate $\mathbbm{B}(\mathcal{\tilde H})$.

At last, we prove the convergence of the characterization of an $n$-partite distribution $P$ with an underlying PPT mixture, which shows that the hierarchy given by Eq.~\eqref{eq:ppt_genuine_tsirelson} is complete.

\begin{proposition}
If an $n$-partite distribution~$P$ satisfies $\chi^{(\ell)}[P,u]\! =\! \sum \chi^{(\ell)}[P_m\!,u_m]$~with $\chi^{(\ell)}[P_m,u_m]= \chi^{(\ell)}[P_m,u_m]^{T_{\bar m}} \!=\! \chi^{(\ell)}[P_m,u_m]^{T} \geq 0$ via suitable $P_m,u_m$ for all levels $\ell$, then it has a (potentially infinite-dimensional) commutative PPT mixture quantum representation.
\end{proposition}

\begin{proof}
From the remark after Proposition~\ref{Proposition:PPT} we know that the conditions on $\chi^{(\ell)}[P_m,u_m]$ for all $\ell$ show that there exists an (un-normalized) distribution $P_m$ which has a commutative PPT quantum representation for the bipartition $m$. Let us refer to this representation as the state $\sigma_m$ and measurements $A_i^{1|m},A_j^{2|m},\ldots,A_k^{n|m}$. Note that here, because the probabilities $P_m$ are not normalized but we rather have $\sum_m P_m = P$, the state $\rho_m$ satisfies $\tr(\rho_m)=\chi[P_m,u_m]_{\tr}$. 

From this one can construct a commutative PPT mixture representation for the given $n$-partite distribution $P$ by using $\mathcal{H} = \oplus_m \mathcal{H}_m$, $\rho = \oplus_m \sigma_m$ and $A_i^s = \oplus_m A_i^{s|m}$ for each system $s$. This construction directly ensures positivity and that the measurements commute. A commutative PPT mixture decomposition of the state is then given by $\rho = \sum_m \rho_m$ with (un-normalized) states $\rho_m=\Pi_m \rho \Pi_m$ and $\Pi_m$ denoting the projector onto $\mathcal{H}_m$. That these states $\rho_m$ indeed satisfy the relations given by Eq.~(\ref{eq:ppt_comm_multi}) follows from the corresponding identities on $\sigma_m$ since $\rho_m A_i^s=\sigma_m A_i^{s|m}$ for each system $s$ and similarly for products of more operators. This concludes the proof. 
\end{proof}

\end{document}